\newif\ifpublic
\publictrue 

\documentclass[11pt]{article}

\usepackage{ttpalatino}
\usepackage[margin=1in]{geometry}

\usepackage{amsmath}
\usepackage{amssymb}
\usepackage{amsfonts}
\usepackage{graphicx}
\usepackage{fullpage}
\usepackage{pdfsync}
\usepackage{underscore}  
\usepackage{amsthm}
\usepackage{xcolor}

\usepackage[bookmarks,colorlinks,breaklinks]{hyperref}  
\hypersetup{linkcolor=blue,citecolor=blue,filecolor=blue,urlcolor=blue} 
\newcommand{\lref}[2][]{\hyperref[#2]{#1~\ref*{#2}}}
\renewcommand{\eqref}[1]{\hyperref[#1]{(\ref*{#1})}}
\numberwithin{equation}{section}

\theoremstyle{plain}
\newtheorem{lem}{Lemma}[section]
\newtheorem{theorem}[lem]{Theorem}
\newtheorem{lemma}[lem]{Lemma}
\newtheorem{fact}[lem]{Fact}

\newtheorem{claim}[lem]{Claim}

\newtheorem{definition}[lem]{Definition}

\theoremstyle{definition}
\newtheorem{remark}[lem]{Remark}

\ifpublic
\usepackage[disable]{todonotes}
\else
\usepackage[colorinlistoftodos]{todonotes}
\fi

\newcommand{\ph}[1]{\todo[color=red!100!green!33, size=\footnotesize]{ph: #1}}

\DeclareMathOperator{\supp}{support}
\DeclareMathOperator{\poly}{poly}
\DeclareMathOperator{\rank}{rank}
\DeclareMathOperator{\kernal}{ker}
\DeclareMathOperator{\eval}{eval}
\DeclareMathOperator*{\E}{\mathbb{E}}

\newcommand{\etal}{{\em et.~al.}}
\newcommand{\F}{\mathbb{F}}
\newcommand{\C}{\mathbb{C}}
\newcommand{\N}{\mathbb{N}}

\newcommand{\Pe}{\mathsf{P}}

\newcommand{\NP}{\mathsf{NP}}
\newcommand{\DTIME}{\mathsf{DTIME}}
\newcommand{\TSAT}{\mathsf{3SAT}}
\newcommand{\NAE}{\mathsf{NAE}}
\newcommand{\OPT}{\text{OPT}}
\newcommand{\LC}{\textsf{LC}}

\renewcommand{\epsilon}{\varepsilon}
\renewcommand{\phi}{\varphi}
\newcommand{\Far}{\text{\sc far}}
\newcommand{\Near}{\text{\sc near}}

\newcommand{\prob}[2]{\Pr_{#1}\left[#2\right]}
\newcommand{\avg}[2]{\mathop{\mathbb{E}}_{#1}\left[#2\right]}

\newcommand{\field}{\mathbb{F}}

\newcommand{\naturals}{\mathbb{N}}

\newcommand{\mc}[1]{\mathcal{#1}}

\newcommand {\email} [1] {Email: \texttt{#1}.}

\begin{document}
\title{Super-polylogarithmic hypergraph coloring hardness via low-degree long codes}
\author{Venkatesan Guruswami\thanks{Computer Science Department, Carnegie
    Mellon University, USA. Research supported in part by a Packard
    Fellowship, US-Israel BSF grant number 2008293, and the US
    National Science Foundation Grant No. CCF-1115525. \email{guruswami@cmu.edu}}
\and
Prahladh Harsha\thanks{Tata Institute of Fundamental Research,
  India. Part of the work was done while the author was visiting the
  Simons Institute for Theory of Computing. \email{prahladh@tifr.res.in}}
\and
Johan H\aa stad\thanks{KTH Royal Institute of Technology, Sweden. Part of the work was done while the author was visiting the
  Simons Institute for Theory of Computing. Partly supported by ERC
  grant 226203. \email{johanh@kth.se}}
\and Srikanth Srinivasan\thanks{Department of Mathematics, IIT Bombay,
  India. \email{srikanth@math.iitb.ac.in}}
\and Girish Varma\thanks{Tata Institute of Fundamental Research,
  India. Supported by Google India under the Google India PhD
  Fellowship Award. \email{girishrv@tifr.res.in}}
}
\begin{titlepage}

\maketitle
\thispagestyle{empty}

\vspace{-0.15in}
\begin{abstract}
We prove improved inapproximability results for hypergraph coloring using the
low-degree polynomial code (aka, the ``short code'' of
Barak~\etal~[FOCS 2012]) and the techniques proposed by
Dinur and Guruswami~[FOCS 2013] to incorporate this code for inapproximability results.

In particular, we prove quasi-NP-hardness of the following problems on $n$-vertex hypergraphs:
\begin{itemize}
\item Coloring a 2-colorable 8-uniform hypergraph with
  $2^{2^{\Omega(\sqrt{\log \log n})}}$ colors.\label{item28}
\item Coloring a 4-colorable 4-uniform hypergraph with
  $2^{2^{\Omega(\sqrt{\log \log n})}}$ colors.\label{item44}
\item Coloring a 3-colorable 3-uniform hypergraph with
  $(\log n)^{\Omega(1/\log\log\log n)}$ colors.\label{item33}
\end{itemize}
In each of these cases, the hardness results obtained are (at least)
exponentially stronger than what was previously known for the
respective cases. In fact, prior to this result, $(\log n)^{O(1)}$
colors was the strongest quantitative bound on the number of colors
ruled out by inapproximability results for $O(1)$-colorable
hypergraphs.

The fundamental bottleneck in obtaining coloring inapproximability
results using the low-degree long code was a multipartite structural
restriction in the PCP construction of Dinur-Guruswami. We are able
to get around this restriction by simulating the multipartite
structure implicitly by querying just one partition (albeit requiring
8 queries), which yields our result for 2-colorable 8-uniform
hypergraphs.  The result for 4-colorable 4-uniform hypergraphs is
obtained via a ``query doubling" method exploiting additional
properties of the $8$-query test. For 3-colorable 3-uniform
hypergraphs, we exploit the ternary domain to design a test with an
{\em additive} (as opposed to multiplicative) noise function, and
analyze its efficacy in killing high weight Fourier coefficients via
the pseudorandom properties of an associated quadratic form. The
latter step involves extending the key algebraic ingredient of
Dinur-Guruswami concerning testing binary Reed-Muller codes to the
ternary alphabet.
\end{abstract}

\end{titlepage}

\section{Introduction}

The last two decades have seen tremendous progress in understanding
the hardness of approximating constraint satisfaction
problems. Despite this progress, the status of approximate coloring of
constant colorable (hyper)graphs is not resolved and in fact, there is
an exponential (if not doubly exponential) gap between the best known
approximation algorithms and inapproximability results. The current
best known approximation algorithms require at least $n^{\Omega(1)}$
colors to color a constant colorable (hyper)graph on $n$ vertices
while the best inapproximability results only rule out at best $(\log
n)^{O(1)}$ (and in fact, in most cases, only $o(\log n)$) colors.

Given this disparity between the positive and negative results, it is
natural to ask why current inapproximability techniques get stuck at
the $\poly\log n$ color barrier. The primary bottleneck in going past
polylogarithmic colors is the use of the {\em long code}, a
quintessential ingredient in almost all tight inapproximability results,
since it was first introduced by Bellare, Goldreich and
Sudan~\cite{BellareGS1998}. The long code, as the name suggests, is
the most redundant encoding, wherein a $n$-bit Boolean string $x$ is
encoded by a $2^{2^n}$-bit string which consists of the evaluation of
all Boolean functions on $n$ bits at the point $x$. It is this doubly
exponential blowup of the long code which prevents the coloring
inapproximability to go past the $\poly \log n$ barrier. Recently,
Barak~\etal~\cite{BarakGHMRS2012}, while trying to understanding the
tightness of the Arora-Barak-Steurer algorithm for unique games,
introduced the {\em short code}, also called the {\em low-degree long
  code}~\cite{DinurG2013}. The low-degree long code is a puncturing of
the long code in the sense, that it contains only the evaluations of
low-degree functions (opposed to all
functions). Barak~\etal~\cite{BarakGHMRS2012} introduced the
low-degree long code to prove exponentially stronger integrality gaps
for Unique Games, and construct small set expanders whose Laplacians
have many small eigenvalues,

Being a derandomization of the long code, one might hope to use the low-degree long code as a more size-efficient surrogate for the long code in inapproximability results. In fact, Barak~\etal~\cite{BarakGHMRS2012} used
it obtain a more efficient version of the KKMO alphabet
reduction~\cite{KhotKMO2007} for Unique Games. However, using the low-degree long code towards improved reductions from Label Cover posed some challenges related to folding, and incorporating noise without giving up perfect completeness (which is crucial for results on coloring). Recently, Dinur and Guruswami~\cite{DinurG2013}
introduced a very elegant set of techniques to adapt the long code
based inapproximability results to low-degree long codes. Using these
techniques, they proved (1) improved inapproximability results for
gap-$(1,\frac{15}{16}+\epsilon)$-4SAT for $\epsilon =
\exp(-2^{\Omega(\sqrt{\log \log N})})$ (long code based reductions
show for $\epsilon = 1/\poly\log N$) and (2) hardness for a variant of
approximate hypergraph coloring, with a gap of 2 and
$\exp(2^{\Omega(\sqrt{\log\log N})})$ number of colors (where $N$ is
the number of vertices). It is to be noted that the latter is the
first result to go beyond the logarithmic barrier for a coloring-type
problem. However, the Dinur-Guruswami~\cite{DinurG2013} results do not
extend to standard (hyper)graph coloring hardness due to a
multipartite structural bottleneck in the PCP construction, which we
elaborate below.

As mentioned earlier, the two main contributions of Dinur-Guruswami~\cite{DinurG2013}
are (1) folding mechanism over the low-degree long code and (2) noise
in the low-degree polynomials. The results of
Bhattacharyya~\etal~\cite{BhattacharyyaKSSZ2010} and
Barak~\etal~\cite{BarakGHMRS2012} suggest that the product of $d$
linearly independent affine functions suffices to work as noise for
the low-degree long code setting (with degree = $d$) in the sense that
it attenuates the contribution of large weight Fourier
coefficients. However, this works only for PCP tests with imperfect
completeness. Since approximate coloring results require perfect
completeness, Dinur and Guruswami~\cite{DinurG2013} inspired by the
above result, develop a noise function which is the product of two
random low-degree polynomials such that the sum of the degrees is at
most $d$. This necessitates restricting certain functions in the PCP
test to be of smaller degree which in turn requires the PCP tests to
query two types of tables -- one a low-degree long code of degree $d$
and another a low-degree long code of smaller degree. Though the
latter table is a part of the former, a separate table is needed since
otherwise the queries will be biased to the small degree portion of
the low-degree long code. This multipartite structure is what
precludes them from extending their result for standard coloring
results. (Clearly, if the query of the PCP tests straddles two 
tables, then the associated hypergraph is trivially 2-colorable.)

\subsection{Hypergraph coloring results}

In this work, we show how this multipartite structural restriction can
be overcome, thus yielding (standard) coloring inapproximability
results. The first of our results extends the result of
Dinur-Guruswami~\cite{DinurG2013}: variant of 6-uniform hypergraph
coloring result to a standard hypergraph coloring result, albeit of
larger uniformity, namely 8.

\begin{theorem}[2-colorable 8-uniform hypergraphs]
\label{thm:2c8u}
Assuming $\NP \not\subseteq \DTIME(n^{2^{O(\sqrt{\log \log n})}})$, there is no polynomial time algorithm which, when given as input an $8$-uniform hypergraph $H$ on $N$ vertices can distinguish between the following:
\begin{itemize}
\item $H$ is $2$ colorable,
\item $H$ has no independent set of size $N/2^{2^{O(\sqrt{\log \log N})}}$.
\end{itemize}
\end{theorem}

This result is obtained using the framework of Dinur-Guruswami~\cite{DinurG2013} by
showing that the two additional queries can be used to simulate
queries into the smaller table via queries into the larger table.

We note that prior to this result, $(\log N)^{\Omega(1)}$ colors was the
strongest quantitative bound on hardness for hypergraph coloring: Khot
obtained such a result for coloring 7-colorable 4-uniform
hypergraphs~\cite{Khot2002c} while Dinur and
Guruswami~\cite{DinurG2013} obtained a similar (but incomparable)
result for 2-colorable 6-uniform hypergraphs both using the long code.

We observe that the 8-query PCP test used in the above
inapproximability result has a stronger completeness guarantee than
required to prove the above result: the 8 queries of the Not-All-Equal
($\NAE$) PCP test, say
$e_1,e_2,e'_1,e'_2,e_3,e_4,e'_3,e'_4$ in the completeness case satisfy
$$\NAE(A(e_1), A(e_2)) \vee \NAE(A(e'_1), A(e'_2)) \vee \NAE(A(e_3), A(e_4))
\vee \NAE(A(e'_3), A(e'_4))$$ which is stronger than the required
$$\NAE(A(e_1), A(e_2),A(e'_1), A(e'_2),A(e_3), A(e_4),A(e'_3),
A(e'_4)).$$ Furthermore, for each $i$, the queries $e_i$ and $e'_i$
appear in the same table. This lets us perform the following
``doubling of queries'': each location is now indexed by a pair of
queries, e.g., $(e_1,e'_1)$ and is expected to return 2 bits which are
the answers to the two queries respectively. The stronger completeness
property yields a 4-query $\NAE$ PCP test over an alphabet of size 4
with the completeness property,
$$\NAE(B(e_1,e'_1),B(e_2,e'_2)) \vee \NAE(B(e_3,e'_3),B(e_4,e'_4)),$$
which suffices for the completeness for proving inapproximability
results for 4-colorable 4-uniform hypergraphs.  We show that the
soundness analysis also carries over to yield the following hardness
for 4-colorable 4-uniform hypergraphs.

\begin{theorem}[4-colorable 4-uniform hypergraphs]
\label{thm:4u4c}
Assuming $\NP \not\subseteq \DTIME(n^{2^{O(\sqrt{\log \log n})}})$, there is no polynomial time algorithm which, when given as input a $4$-uniform hypergraph $H$ on $N$ vertices can distinguish between the following:
\begin{itemize}
\item $H$ is $4$ colorable,
\item $H$ has no independent set of size  $N/2^{2^{O(\sqrt{\log \log N})}}$.
\end{itemize}
\end{theorem}

We remark that the doubling method, mentioned above, 
when used in the vanilla long code setting (as opposed to low-degree
long code setting) already yields the following inapproximability: it is
quasi-NP-hard to color a 4-colorable 4-uniform hypergraph with
$(\log N)^{\Omega(1)}$ colors. This result already improves upon the above mentioned
result of Khot~\cite{Khot2002c} for 7-colorable 4-uniform
hypergraphs. Another feature of the doubling method is that although
the underlying alphabet is of size 4, namely $\{0,1\}^2$, it suffices
for the soundness analysis to perform standard Fourier analysis over
$\F_2$.

In the language of covering complexity\footnote{The covering number of
  a CSP is the minimal number of assignments to the vertices so that
  each hyperedge is covered by at least one assignment},
(the proof of) \lref[Theorem]{thm:4u4c} demonstrates a Boolean 4CSP for which it
is quasi-NP-hard to distinguish between covering number of 2
vs. $\exp(\sqrt{\log \log N})$. The previous best result for a Boolean
4CSP was 2 vs. $\log\log N$, due to Dinur and Kol~\cite{DinurK2013}.

We then ask if we can prove coloring inapproximability for even smaller
uniformity, i.e., 2 and 3 (graphs and 3-uniform hypergraphs
respectively). We show that we can use a different noise function over
$\F_3$ to obtain the following inapproximability result for
3-colorable 3-uniform hypergraphs.

\begin{theorem}[3-colorable 3-uniform hypergraphs]
\label{thm:3u3c}
Assuming $\NP \notin \DTIME(n^{2^{O(\log \log n / \log \log \log n)}})$, there is no polynomial time algorithm which, when given as input a $3$-uniform hypergraph $H$ on $N$ vertices can distinguish between the following:
\begin{itemize}
\item $H$ is $3$ colorable.
\item $H$ has no independent set of size $N/2^{O(\log \log N / \log \log \log N)}$. 
\end{itemize}
\end{theorem}
Prior to this result, the best inapproximability result for
O(1)-colorable 3-uniform hypergraphs were as follows:
Khot~\cite{Khot2002b} showed that it is quasi-NP-hard to color a
3-colorable 3-uniform hypergraphs with $(\log\log N)^{1/9}$
colors and Dinur, Regev and Smyth~\cite{DinurRS2005} showed that it is
quasi-NP-hard to color a 2-colorable 3-uniform hypergraphs with
$(\log\log N)^{1/3}$ colors (observe that $2^{O(\log \log N /
  \log \log \log N)}$ is exponentially larger than $(\log \log
N)^{\Omega(1)}$). For 2-colorable 3-uniform hypergraphs, the
result of Dinur~\etal~\cite{DinurRS2005} only rules out colorability by
$(\log \log N)^{\Omega(1)}$, while a recent result due to Khot and
Saket~\cite{KhotS2014} shows that it is hard to find a $\delta
N$-sized independent set in a given $N$-vertex 2-colorable 3-uniform
hypergraph assuming the $d$-to-$1$ games conjecture.  Our improved
inapproximability result is obtained by adapting Khot's proof to the
low-degree long code using the new noise function over $\F_3$. We
remark that this result is not as strong as the previous two
($2^{O(\log \log N / \log \log \log N)}$ instead of
$2^{2^{O(\sqrt{\log \log N})}}$) as for 3-uniform hypergraphs, the
starting point is a multilayered smooth label cover instance instead
of just label cover, which causes a blowup in size and a corresponding
deterioration in the parameters. 

\subsection{Low-degree long code analysis via Reed-Muller testing}

One of the key contributions of Barak~\etal~\cite{BarakGHMRS2012} was
the discovery of a connection between Reed-Muller testing and the
analysis of the low-degree long code. In particular, they showed the
following. Let $\Pe^n_d$ set of degree $d$ polynomials on $n$
variables over $\F_2$. For functions $\beta, g : \F_2^n \to\F_2$, let
$\chi_\beta(g)= (-1)^{\sum_{x \in \F_2^n}
  \beta(x)g(x)}$. Barak~\etal~oberved that if $\beta$ is far from the
set $P^n_{n-d-1}$ of degree $n-d-1$ polynomials, then one can bound
the expectation $|\E_\mu\left[\chi_\beta(\eta)\right]|$ for a random
low-weight $\eta$ using a powerful result on Reed-Muller testing over
$\F_2$ due to Bhattacharyya~\etal~\cite{BhattacharyyaKSSZ2010}. This
demonstrates that the noise function $\eta$ attenuates the
contribution of high-order Fourier coefficients and is thus useful in
the low-degree long code analysis. However, this noise $\eta$ has
imperfect completeness and Dinur-Guruswami had to prove a new result
on Reed-Muller testing over $\F_2$ to construct a noise function that
allows for perfect completeness. They showed that if $\beta$ is
$2^{d/2}$-far from $\Pe^{n}_{n-d-1}$, then
$\E_{g\in\Pe^{n}_{d/4}}\left|\E_{h\in \Pe^{n}_{3d/4}}
  [\chi_\gamma(gh)]\right|$ was doubly exponentially small in $d$ (see
\lref[Theorem]{thm:DG} for a fomal statement).  This allowed them to
extend some of the long code based inapproximability with perfect
completeness to the low-degree long code setting. Tests based on the
above property need to access functions of different degree (e.g., $g,
gh$ in the above discussion) and this results in a multipartite structure in the
low-degree long code tables of \cite{DinurG2013}. The results for
2-colorable 8-uniform hypergraphs and 4-uniform 4-colorable
hypergraphs are obtained using the above result of \cite{DinurG2013}.

For the case of 3-uniform 3-colorable hypergraphs, we observe that if
we extend the alphabet to ternary (i.e., $\F_3$ instead of $\F_2$), we
can design a noise function that has both perfect completeness and
does not result in a multipartite structural restriction. Let
$\Pe^n_d$ now denote the set of degree $d$ polynomials on $n$
variables over $\F_3$. We show that if $\beta: \F_3^n\to\F_3$ is
$3^{d/2}$-far from $\Pe^{n}_{2n-2d-1}$, then $\left|\E_{p\in
    \Pe^{n}_{d}} [\chi_\beta(p^2)]\right|$ is doubly exponentially
small in $d$. This is proved by showing the following pseudorandom
property of the associated quadratic form $Q^\beta$ defined as
$Q^\beta:= \sum_{x\in \F_3^n} \beta(x) \cdot \eval(x)\eval(x)^T$ where
$\eval(x)$ is the column-vector of evaluation of all degree $d$
monomials at the point $x$. If the distance of $\beta$ from
polynomials of degree $2n-2d-1$, denoted by $\Delta_d(\beta)$ is at
least $3^{d/2}$, then the rank of the matrix $Q(\beta)$ is exponential
in $d$ and is otherwise equal to the distance $\Delta_d(\beta)$. This
rank bound is proved along the lines of \cite{DinurG2013} using the
Reed-Muller tester analysis of Haramaty, Shpilka and
Sudan~\cite{HaramatySS2013} over general fields instead of the
Bhattacharyya~\etal~\cite{BhattacharyyaKSSZ2010} analysis over $\F_2$.
\ph{Add related work section, explain almost-colorable results}

\subsection*{Organization}
We start with some preliminaries in
\lref[Section]{sec:prelims}. Theorems~\ref{thm:2c8u}, \ref{thm:4u4c},
and \ref{thm:3u3c} are proved in Sections~\ref{sec:2c8u},
\ref{sec:4u4c}, and \ref{sec:3u3c} respectively. The proof of the
latter theorem requires a technical claim about low-degree polynomials
over $\field_3$, which we prove in
\lref[Section]{sec:test-analysis}. 

\section{Preliminaries}\label{sec:prelims}

\subsection{Label cover}
All our reductions start from an appropriate instance of the label
cover problem, bipartite or multipartite.  
A bipartite label
cover instance consists of a bipartite graph $G=(U,V,E)$, label sets
$\Sigma_U,\Sigma_V$, and a set of projection constraints
$\Pi= \{\pi_{uv}:\Sigma_U\rightarrow \Sigma_V| (u,v) \in E\}$.
We consider label cover instances obtained
from $\TSAT$ instances in the following natural manner.
\begin{definition}[$r$-repeated label cover]
\label{def:label-cover}
Let $\phi$ be a $\TSAT$ instance with $X$ as the set of variables and $C$ the set of clauses. The $r$-repeated bipartite label cover instance $I(\phi)$ is specified by:
\begin{itemize}
\item A graph $G:=(U,V,E)$, where $U:=C^r, V:=X^r$.
\item $\Sigma_U := \{0,1\}^{3r},\Sigma_V := \{0,1\}^r$. 
\item There is an edge $(u,v) \in E$ if the tuple of variables $v$ can be obtained from the tuple of clauses $u$ by replacing each clause by a variable in it.
\item The constraint $\pi_{uv}:\{0,1\}^{3r}\rightarrow \{0,1\}^{r}$ is simply the projection of the assignments on $3r$ variables in all the clauses in $u$ to the assignments on the $r$ variables in  $v$. 
\item For each $u$ there is a set of $r$ functions
  $\{f^u_i:\{0,1\}^{3r} \rightarrow \{0,1\} \}_{i=1}^r$ such that
  $f^u_i(a)=0$ iff the assignment $a$ satisfies the $i$th clause in
  $u$. Note that $f^u_i$ depends only on the $3$ variables in the
  $i$th clause.
\end{itemize}
A labeling $L_U:U\rightarrow \Sigma_U,L_V:V\rightarrow \Sigma_V$
satisfies an edge $(u,v)$ iff $\pi_{uv}(L_U(u))=L_V(v)$ and $L_U(u)$
satisfies all the clauses in $u$. Let $\OPT(I(\phi))$ be the maximal
fraction of constraints that can be satisfied by any labeling.
\end{definition}
The following theorem is obtained by applying Raz's parallel repetition theorem~\cite{Raz1998} with $r$ repetitions on hard instances of $\mathsf{MAX}$-$\TSAT$ where each variable occurs the same number of times~\cite{Feige1998}.
\begin{theorem}
\label{thm:label-cover}
There is an algorithm which on input a $\TSAT$ instance $\phi$ and
$r\in \N$ outputs an $r$-repeated label cover instance $I(\phi)$ in
time  $n^{O(r)}$ with the following properties. 
\begin{itemize}
\item If $\phi \in \TSAT$, then $\OPT(I(\phi))=1$.
\item If $\phi \notin \TSAT$, then $\OPT(I(\phi)) \leq 2^{-\epsilon_0
    r}$ for some universal constant $\epsilon_0\in (0,1)$. 
\end{itemize}
Moreover, the underlying graph $G$ is both left and right regular.
\end{theorem}

\paragraph{Multilayered smooth label cover:}
For our hardness results for $3$-uniform $3$-colorable hypergraphs, we
need a multipartite version of label cover, satisfying a smoothness
condition.
\begin{definition}[smoothness]
Let $I$ be a bipartite label cover instance specified by $\left((U,V,E),\Sigma_U,\Sigma_V,\Pi\right)$. Then $I$ is $\eta$-\emph{smooth} iff for every $u \in U$ and two distinct labels $a,b \in \Sigma_U$
$$\Pr_v[ \pi_{uv}(a) = \pi_{uv}(b) ] \leq \eta,$$
where $v$ is a random neighbour of $u$.
\end{definition}
\begin{definition}[$r$-repeated $\ell$-layered $\eta$-smooth label cover]\label{def:multilayer}
Let $T:=\lceil\ell/\eta\rceil$ and $\phi$ be a $\TSAT$ instance with
$X$ as the set of variables and $C$ the set of clauses. The
$r$-repeated $\ell$-layered $\eta$-smooth label cover instance
$I(\phi)$ is  specified by:
\begin{itemize}
\item An $\ell$-partite graph with vertex sets $V_0, \cdots
  V_{\ell-1}$. Elements of $V_i$ are tuples of the form $(C',X')$
  where $C'$ is a set of $(T+\ell - i)r$ clauses and $X'$ is a a set
  of $ir$ variables.
\item $\Sigma_{V_i} := \{0,1\}^{m_i}$ where $m_i:={3(T+\ell - i)r +ir}$ which corresponds
  to all Boolean assignments to the clauses and variables
  corresponding to a vertex in layer $V_i$.
\item  For $0 \leq i < j < \ell$, $E_{ij} \subseteq V_i \times V_j$
  denotes the set of edges between layers $V_i$ and $V_j$. For $v_i
  \in V_i, v_j \in V_j$, there is an edge $(v_i,v_j) \in E_{ij}$ iff
  $v_j$ can be obtained from $v_i$ by replacing some $(j-i)r$ clauses
  in $v_i$ with variables occurring in the clauses respectively.
\item The constraint $\pi_{v_i v_j}$ is the projection of assignments for clauses and variables in $v_i$ to that of $v_j$. 
\item For each $i <\ell$, $v_i \in V_i$, there are $(T+\ell - i)r$ functions $f_j^{v_i}:\{0,1\}^{3(T+\ell - i)r +ir}\rightarrow \{0,1\}$, one for each clause $j$ in $v_i$ such that $f_j^{v_i}(a)=0$ iff $a$ satisfies the clause $j$. This function only depends on the $3$ coordinates in $j$.
\end{itemize}
Given a labeling $L_i:V_i\rightarrow \Sigma_{V_i}$ for all the vertices, an edge $(v_i,v_j) \in E_{ij}$ is satisfied iff $L_i(v_i)$ satisfies all the clauses in $v_i$, $L_j(v_j)$ satisfies all the clauses in $v_j$ and $\pi_{v_i v_j}(L_i(v_i)) = L_j(v_j)$. Let $\OPT_{ij}(I(\phi))$ be the maximum fraction of edges in $E_{ij}$ that can be satisfied by any labeling.
\end{definition}
The following theorem was proved by Dinur~\etal~\cite{DinurGKR2005} in
the context of hypergraph vertex cover inapproximability (also see \cite{DinurRS2005}).
\begin{theorem}
\label{thm:layered-label-cover} There is an algorithm which on input a $\TSAT$ instance $\phi$ and $\ell,r\in \N, \eta \in [0,1)$ outputs a $r$-repeated $\ell$-layered $\eta$-smooth label cover instance $I(\phi)$ in time  $n^{O((1+1/\eta)\ell r)}$ with the following properties.
\begin{enumerate}
\item $\forall~ 0\leq i < j < \ell$, the bipartite label cover
  instance on
  $I_{ij}=\left((V_i,V_j,E_{ij}),\Sigma_{V_i},\Sigma_{V_j},\Pi_{ij}\right)$
  is $\eta$-smooth.
\item For $1<m<\ell$,  any $m$ layers $0\leq i_1< \cdots <i_m\leq
  \ell-1$, any $S_{i_j} \subseteq V_{i_j}$ such that $|S_{i_j}| \geq
  \frac{2}{m}|V_{i_j}|$, there exists distinct ${i_j}$ and ${i_{j'}}$
  such that the fraction of edges between $S_{i_j}$ and $S_{i_{j'}}$ relative to $E_{i_ji_{j'}}$ is
  at least $1/m^2$. 
\item If $\phi \in \TSAT$, then there is a labeling for $I(\phi)$ that satisfies all the constraints.
\item If $\phi \notin \TSAT$,  then 
$$\OPT_{i,j}(I(\phi)) \leq 2^{-\Omega(r)}, \quad \forall 0\leq i < j
\leq \ell.$$
\end{enumerate}
\end{theorem}

\subsection{Low-degree long code}\label{sec:ldlc}

Let $\F_p$ be the finite field of size $p$ where $p$ is a prime. The
results in this section apply when $p=2,3$. The choice of
$p$ will be clear from context and hence the dependence of $p$ on the
quantities defined will be omitted. Let $\Pe^n_d$ be the set of degree
$d$ polynomials on $n$ variables over $\F_p$. Let $\mathfrak F_n :=
\Pe^n_{(p-1)n}$. Note that $\mathfrak F_n$ is the set of all
functions from $\F_p^n$ to $\F_p$. $\mathfrak F_n$ is a $\F_p$-vector
space of dimension $p^n$ and $\Pe^n_d$ is its subspace of dimension
$n^{O(d)}$. The Hamming distance between $f$ and $g \in \mathfrak F_n$,
denoted by $\Delta(f,g)$, is the number of inputs on which $f$ and $g$
differ. When $S \subseteq \mathfrak F_n$, $\Delta(f, S) := \min_{g\in
  S} \Delta(f,g)$. 
We say $f$ is $\Delta$-far from $S$ if 
$\Delta(f,S) \geq \Delta$ and $f$ is $\Delta$-close to $S$ otherwise. Given $f,g, \in \mathfrak F_n$, the dot product
between them is defined as $\langle f,g \rangle := \sum_{x \in \F_p^n}
f(x)g(x)$. For a subspace $S \subseteq \mathfrak F_n$, the dual
subspace is defined as $S^{\perp} := \{ g \in \mathfrak F_n : \forall f
\in S, \langle g,f \rangle = 0 \}$. The following theorem relating
dual spaces is well known.
\begin{lemma} 
\label{lem:dual}$(\Pe^{n}_{d})^\perp = \Pe^n_{(p-1)n-d-1}$.
\end{lemma}
\noindent We need the following Schwartz-Zippel-like Lemma for degree $d$ polynomials.
\begin{lemma}[Schwartz-Zippel lemma~{\cite[Lemma~3.2]{HaramatySS2013}}]
\label{lem:SZ}
Let $f\in \F_p[x_1,\cdots,x_n]$ be a non-zero polynomial of degree at most $d$ with individual degrees at most $p-1$. Then $\prob{a\in \F_p^n}{f(a)\neq 0} \geq p^{-d/p-1}$.
\end{lemma}
\noindent We now define the low-degree long code (introduced as the short code
by Barak~\etal~\cite{BarakGHMRS2012} in the $\F_2$ case).
\begin{definition}[low-degree long code]
For $a \in \F_p^n$, the degree $d$ long code for $a$ is a function $\LC_d(a):\Pe^n_{d} \rightarrow \F_p$ defined as
$$\LC_d(a)(f) := f(a).$$
\end{definition}
\noindent Note that for $d=(p-1)n$, this matches with the definition of the
original long code over the alphabet $\F_p$. 
\begin{definition}[characters]
A character of $\Pe^n_d$ is a function $\chi:\Pe^n_{d} \rightarrow \C$ such that
$$\chi(0) = 1 \text{ and } \forall f,g \in \Pe^n_d,~ \chi(f+g) = \chi(f)\chi(g).$$
\end{definition}
\noindent The following lemma lists the basic properties of characters.
\begin{lemma}
\label{lem:fourier}
Let $\{1,\omega,\cdots,\omega^{p-1}\}$ be the $p$th roots of unity and
for $\beta \in \mathfrak F_n, f \in \Pe^n_d$, $\chi_\beta(f) :=
\omega^{\langle \beta, f \rangle}$.
\begin{itemize}
\item The characters of $\Pe^n_d$ are $\{ \chi_\beta : \beta \in  \mathfrak F_n\}$. 
\item For any $\beta,\beta' \in \mathfrak F_n$, $\chi_\beta = \chi_\beta'$ if and only if  $\beta-\beta' \in (\Pe^{n}_d)^\perp$.
\item For $\beta \in (\Pe^{n}_d)^\perp$, $\chi_\beta$ is the constant $1$ function.
\item\label{item:minsup} $\forall \beta, \exists \beta'$ such that $\beta-\beta' \in
  (\Pe^{n}_d)^\perp$ and $|\supp(\beta') | = \Delta(\beta,
  (\Pe^{n}_d)^\perp)$ (i.e., the constant $0$ function is (one of) the
  closest function to $\beta'$ in $(\Pe^{n}_d)^\perp$). We call such a
  $\beta'$ a minimum support function for the coset $\beta + (\Pe^n_d)^\perp$.
\item Characters forms an orthonormal basis for the vector space of functions from $\Pe^n_d$ to $\C$, under the inner product  $\langle A, B\rangle := \E_{f \in \Pe^n_d} \left[A(f)\overline{B(f)}\right]$
\item Any function $A:\Pe^n_d \rightarrow \C$ can be uniquely decomposed as
$$A(f) = \sum_{\beta \in \Lambda^n_d} \widehat{A}(\beta) \chi_\beta(f)
\text{ where } \widehat{A}(\beta) := \E_{g \in \Pe^n_d} \left[A(g)
  \overline{\chi_\beta(g)}\right],$$ 
where $\Lambda^n_d$ is the set of minimum support functions, one for
each of the cosets in $\mathfrak F_n/(\Pe^{n}_d)^\perp$, with ties broken arbitrarily.
\item Parseval's identity: For any function $A:\Pe^n_d \rightarrow\C$,
$\sum_{\beta \in \Lambda^n_d} |\widehat A(\beta)|^2 = \E_{f\in \Pe^n_d} [|A(f)|^2].$ In particular, if $A:\Pe^n_d\rightarrow \{1,\omega,\cdots,\omega^{p-1}\}$, $\sum_{\beta\in \Lambda^n_d}|\widehat A(\beta)|^2 =1$.
\end{itemize}
\end{lemma}
\noindent The following lemma relates characters over different
domains related by co-ordinate projections.
\begin{lemma}
\label{lem:char-projection}
Let $m \leq n$ and $\pi:\F_p^n \rightarrow \F_p^m$ be a (co-ordinate)
projection i.e., there exist indices $1 \leq i_i < \cdots < i_m \leq
n$ such that $\pi(x_1,\dots,x_n) = (x_{i_1}, \cdots ,x_{i_m})$. Then for $f \in \Pe^m_d, ~\beta \in
\Pe^n_d$,
$$\chi_\beta(f\circ \pi)= \chi_{\pi_p(\beta)}(f),$$
 where $\pi_p(\beta)(y):= \sum_{x \in \pi^{-1}(y)} \beta(x)$.
\end{lemma}
\begin{proof}
$$\chi_\beta(f\circ \pi) = \omega^{\sum_{x \in \F_3^n} f(\pi(x)) \beta(x)}
= \omega^{\sum_{y \in \F_3^m} f(y) \left(\sum_{x \in \pi^{-1}(y)} \beta(x)\right)}
= \omega^{\sum_{y \in \F_3^m} f(y) \pi_p(\beta)(y)}
= \chi_{\pi_p(\beta)}(f).\qedhere$$
\end{proof}
\noindent Dinur and Guruswami~\cite{DinurG2013} proved the following
theorem about Reed-Muller codes over $\F_2$ using
Bhattacharyya~\etal~\cite{BhattacharyyaKSSZ2010} testing result.
\begin{theorem}[{\cite[Theorem~1]{DinurG2013}}]
\label{thm:DG}
Let $d$ be a multiple of $4$ and $p=2$. If $\gamma\in\mathfrak{F}_{n}$
is $2^{d/2}$-far from $(\Pe^{n}_d)^\perp = \Pe^{n}_{n-d-1}$, then 
$$\E_{g\in\Pe^{n}_{d/4}}\left[\left|\E_{h\in \Pe^{n}_{3d/4}}
    [\chi_\gamma(gh)]\right|\right]\leq 2^{-4\cdot 2^{d/4}}.$$
\end{theorem}
\subsection{Folding over satisfying assignments}
\begin{lemma}
\label{lem:interpol}
Let $d>1$, $X$ be a set of $p^{d}-1$ points in $\F^n_p$ and $f:X \rightarrow \F_p$ an arbitrary function. Then there exists a polynomial $q$ of degree at most $(p-1)d$ such that $q$ agrees with $f$ on all points in $X$.
\end{lemma}
\begin{proof}
  By Lemmas~\ref{lem:dual} and \ref{lem:SZ}, any polynomial in
  $(\Pe^{n}_{(p-1)d})^\perp$ has suppport size at least
  $p^{d}$. Hence, it is possible to interpolate a degree $(p-1)d$
  polynomial through $p^d-1$ points.
\end{proof}
\noindent For any set $S$, a function $A: \Pe^n_{(p-1)d} \rightarrow S$ is said
to be folded over a subspace $J \subseteq \Pe^n_{(p-1)d}$ if $A$ is
constant over cosets of $J$ in $\Pe^n_{(p-1)d}$. 
\begin{fact}
\label{fact:ideallift}
Given a function $A:\Pe^n_{(p-1)d}/J \rightarrow S$ there is a unique
function $A':\Pe^n_{(p-1)d} \rightarrow S$ that is folded over $J$
such that for $g \in \Pe^n_{(p-1)d}, A'(g)  = A(g + J)$. We call $A'$
the lift of $A$.
\end{fact}
\noindent Given $q_1,\cdots, q_k \in \Pe^n_{3(p-1)}$, let
$$ J(q_1,\dots,q_k): = \left\{ \sum_i r_i q_i : r_i \in \Pe^n_{(p-1)(d-3)} \right\}.$$
The following lemma shows that if a function is folded over
$J=J(q_1,\dots,q_k)$, then it cannot have weight on small support
characters that are non-zero on $J$ (this is a generalization of the
corresponding lemma in \cite{DinurG2013} to arbitrary fields).
\begin{lemma}
  \label{lem:goodsupport}Let $\beta \in \mathfrak F_n$ is such that
  $|\supp(\beta)| < p^{d-3}$, and there exists $x \in \supp(\beta)$
  with $q_i(x) \neq 0$ for some $i$. Then if $A:\Pe^{n}_d\rightarrow
  \C$ is folded over $J=J(q_1,\dots,q_k)$, then $\widehat{A}(\beta) =
  0$.
\end{lemma}
\begin{proof}
Construct a polynomial $r$ which is zero at all points in support of $\beta$ except at $x$. From \lref[Lemma]{lem:interpol}, its possible to construct such a polynomial of degree at most $(p-1)(d -3)$. Then we have that $rq_i \in J$ and $\langle\beta,rq_i\rangle \neq 0$. Now
\begin{align*}
\E_h\left[A(h)\chi_\beta(h)\right] &=\frac{1}{p} \E_h\left[ A(h)\chi_\beta(h) + A(h+rq_i)\chi_\beta(h+rq_i)+\cdots + A(h+(p-1)rq_i)\chi_\beta(h+(p-1)rq_i)\right]\\
 &=\frac{1}{p}\E_h\left[ A(h)\chi_\beta(h) + A(h)\chi_\beta(h+rq_i) +\cdots+ A(h)\chi_\beta(h+(p-1)rq_i) \right]\\
&=\frac{1}{p}\E_h\left[  A(h)\chi_\beta(h)(1+\chi_\beta(rq_i)+\cdots+\chi_\beta((p-1)rq_i)) \right]\\
&=0 \qquad \text{[since $\chi_\beta(rq_i) \neq 1$]}\qedhere
\end{align*}
\end{proof}

\section{Correlation with a random square} \label{sec:test-analysis}

In this section, we analyze the quantity
$$\langle \beta, p^2 \rangle,$$
where $p \in \Pe^n_{d}$ is chosen uniformly at random and
$\beta:\F_3^n \rightarrow \F_3$ is a fixed function having distance
exactly $\Delta$ from $(\Pe^{n}_{2d})^\perp=\Pe^n_{2n-2d-1}$. 

Throughout this section, we work over the field $\F_3$. For
$a\in \N^n$, let $|a| := \sum_i a_i$ and $x^a$ denote the monomial
$\prod_{i}x_i^{a_i}$. Over $\F_3$, the individual degrees are at most
$2$ (since $x^3 \equiv x$). Hence, we assume wlog. that the coefficient vector
$a \in \{0,1,2\}^n$. In this notation, $p(x) = \sum_{ |a|\leq d}p_a x^a$ where $p_a$ are chosen independently and
uniformly at random from $\F_3$. For $x \in \F_3^n$, let $e_x$ be the
column vector of evaluation of all degree $d$ monomials at $x$, i.e.,
$e_x := (x^a)_{|a| \leq d}$. Then $p(x) = p^T e_x $ where $p$ is now
thought of as the column vector $(p_a)_{|a| \leq d}$ and hence,
$p^2(x) = (p^Te_x)^2= p^T(e_x e^T_x)p$.
$$ \langle \beta, p^2 \rangle = \sum_x \beta(x)\left(p^Te_xe_x^Tp\right) = p^T\left(\sum_x \beta(x)e_x e_x^T\right)p.$$
We are thus, interested in the quadratic form represented by the
matrix $Q^\beta :=\sum_x \beta(x)e_x^Te_x$. Observe that all $\beta$
belonging to the same coset in $\Pe^n_{2n}/\Pe^n_{2n-2d-1}$ have the
same value for $\langle \beta, p^2 \rangle$ and the matrix
$Q^\beta$. Hence, by \lref[Lemma]{lem:fourier}, we
might wlog. assume that $\beta$ satisifies $\supp(\beta) = \Delta$.
The following lemma (an easy consequence of \cite[Theorem
6.21]{LidlN1997}), shows that it suffices to understand the rank of $Q^\beta$. 
\begin{lemma}
\label{lem:rank-dist}
Let $A$ be a $n\times n$, symmetric matrix with entries from $\F_3$.
The statistical distance of the random variable $p^TAp$ from uniform
is $\exp(-\Omega(\rank(A)))$. 
\end{lemma}
In the next sequence of lemmas, we relate $\rank(Q^\beta)$ to
$\Delta$.  In particular, we show that $\rank(Q^\beta)$ is equal to $\Delta$ if
$\Delta \leq 3^{d/2}$ and is exponential in $d$ otherwise.  Recall
that over $\F_3$, $\Pe^n_{2n}$ is the set of all function from
$\F_3^n$ to $\F_3$ and  $\left(\Pe^{n}_{2d}\right)^\perp =
\Pe^{n}_{2d-2d-1}$. 

\begin{lemma}
\label{lem:rank-upper}
$\rank(Q^\beta) \leq \Delta$.
\end{lemma}
\begin{proof}
By assumption, $\beta$ satisfies $\Delta = \supp(\beta)$. The lemma
follows from that fact that $e_xe_x^T$ are rank one matrices and $Q^\beta=
\sum_x \beta(x) e_xe_x^T$.
\end{proof}

\begin{lemma}
\label{lem:small-dist-rank}
If  $\Delta < 3^{d/2}$, then $\rank(Q^\beta) = \Delta$.
\end{lemma}
\begin{proof}
  By assumption, $\beta$ satisifies $\Delta=\supp(\beta)$ and $Q^\beta = \sum_x \beta(x) e_xe_x^T$. Since
  $(\Pe^n_d)^\perp = \Pe^n_{2n-d-1}$ and any non-zero polynomial with
  degree $2n-d-1$ has support at least $3^{d/2}$
  (\lref[Lemma]{lem:SZ}), any $\lceil 3^{d/2} \rceil - 1$ vectors
  $e_x$ are linearly independent.  In particular, the $\Delta$ vectors
  $e_x$ for $x$ in $\supp(\beta)$ are linearly independent. Consider
  any non-zero $v$ in the kernel of the matrix $Q^\beta$. The linear
  independence of $e_x$'s gives that $e_x^Tv=0$ for all $x \in
  \supp(\beta)$. Hence, the kernel of $Q^\beta$ resides in a
  $\Delta$-codimensional space which implies that $\rank(Q^\beta) =
  \Delta$.
\end{proof}
We conjecture that \lref[Lemma]{lem:small-dist-rank} holds for
larger values of $\Delta$, but for our purposes we only need a lower
bound on the rank when $\Delta \geq 3^{d/2}$.
\begin{lemma}
\label{lem:far-rank}
There exists a constant $d_0$ such that if $d>d_0$ and $\Delta > 3^{d/2}$ then
$\rank(Q^\beta)\geq 3^{d/9}$.
\end{lemma}
\begin{proof}
  The proof of this theorem is similar to the proof of \cite[Theorems
  15,17]{DinurG2013} for the $\F_2$ case and we follow it step by
  step.  Define $B^n_{d,k}(\beta) := \left\{ q \in \Pe^n_k : q\beta
    \in \Pe^n_{2n-2d-1 +k} \right\}$.
\begin{claim}$ \kernal(Q^\beta) = B^n_{d,d}(\beta)$.\end{claim}
\begin{proof}
The matrix $Q^\beta$ satisfies that $Q^\beta(a,b) = \langle \beta ,
x^ax^b \rangle$, for all $a, b \in \{0,1,2\}^n, |a|, |b|\leq d$. Using
this description of $Q^\beta$, we obtain the following description of $\ker(Q^\beta)$.
\begin{align*}
 (h_a)_{|a| \leq d} \in \kernal(Q^\beta)
 &\Longleftrightarrow \forall a : |a| \leq d, &\sum_{b : |b| \leq d} \left\langle
 \beta,x^ax^b\right\rangle h_b =0\\
 &\Longleftrightarrow \forall a : |a| \leq d,& \left\langle \beta,x^a\sum_{b : |b| \leq d}h_bx^b\right\rangle  =0\\
 &\Longleftrightarrow \forall a : |a| \leq d,& \langle \beta x^a,h\rangle  =0\\
 &\Longleftrightarrow \forall q \in \Pe^n_d,& \langle \beta q,h\rangle  =0\\
 &\Longleftrightarrow \forall q \in \Pe^n_d,& \langle \beta h,q\rangle  =0\\ 
 &\Longleftrightarrow \beta h \in \Pe^n_{2n-d-1} & \qquad \qquad\qedhere
\end{align*}
\end{proof}
\noindent Thus to prove \lref[Lemma]{lem:far-rank}, it suffices to
show that $\rank(Q^\beta)=\dim(\Pe^n_d/B^n_{d,d}(\beta)) \geq
3^{d/9}.$ Towards this end, we define 
\begin{equation}\label{eq:Phidk}\Phi_{d,k}(D):= \min_{n > d/2, 
  \beta \in \Pe^n_{2n}: \Delta\left(\beta,\Pe^n_{2n-2d-1}\right) > D} \dim(
\Pe^n_k/ B^n_{d,k}(\beta)).
\end{equation} In terms of $\Phi_{d,k}$, \lref[Lemma]{lem:far-rank} now reduces
to showing that $\Phi_{d,d}(3^{d/2}) \geq 3^{d/9}$. We obtain this
lower bound by recursively bounding this quantity 
. The following serves as the base case of the recursion.
\begin{claim}
\label{claim:base-case}
For $k > 2d$ , $\forall D$, $\Phi_{d,k}(D) = 0$
and for $k \leq 2d$,  $\Phi_{d,k}(1) \geq 1$.
\end{claim}
\begin{proof}
Let $\beta$ be the polynomial which attains the minimum in \eqref{eq:Phidk}.
The first part of the claim follows from the fact that if $k > 2d$ then $B^n_{d,k}(\beta)= \Pe^n_k$. 

Now for the second part. Since $\beta \notin \Pe^n_{2n-2d -1}$, there is a monomial $x^a$ with $|a| \leq 2d$ such that
$$\langle \beta, x^a \rangle \neq 0\Longleftrightarrow \langle \beta x^a, 1 \rangle \neq 0\Longleftrightarrow \beta x^a \notin \Pe^n_{2n -1}.$$
If $|a| \leq k$, $x^a \notin B^n_{d,k}(\beta)$ and we are
done. Otherwise, consider $b$ such that $b \leq a$ coordinate-wise and
$|b| = k$. Suppose $x^b \beta \in \Pe^n_{2n-2d-1 +k}$ then $x^a \beta
\in  \Pe^n_{2n-1}$ which is a contradiction. Hence, $x^b\beta \notin
\Pe^n_{2n-2d-1+k}$ and the second part of the claim follows.
\end{proof}
\noindent For the induction step, we need the following result from Haramaty, Shpilka and Sudan~\cite{HaramatySS2013}.
\begin{claim}[{\cite[Theorems 4.16, 1.7]{HaramatySS2013}}]\label{conj:hyperplane}
There exists a constant $d_0$ such that if $3^5 < \Delta < 3^{d}$, $d>d_0$ where $\beta$
is $\Delta$-far from $\Pe^n_{2n-2d-1}$, then there exists nonzero
$\ell \in \Pe^n_1$ such that $\forall c\in \F_3, \beta|_{\ell = c}$
are $\Delta/27$ far from the restriction of $\Pe^n_{2n-2d-1}$ to affine hyperplanes.
\end{claim}
\noindent See \lref[Appendix]{app:hyperplane} for a proof of
\lref[Claim]{conj:hyperplane} from Theorems~4.16 and 1.7 of \cite{HaramatySS2013}.

\begin{claim}
\label{claim:rec}
If $3^5\leq D \leq 3^{d}$ and $d>d_0$, then
$$\Phi_{d,k}(D) \geq \Phi_{d-1,k}(D/27) + \Phi_{d-1,k-1}(D/27) + \Phi_{d-1,k-2}(D/27).$$
\end{claim}
\begin{proof}
  From \lref[Lemma]{conj:hyperplane}, we get that there exists nonzero
  $\ell \in \Pe^n_1$ such that for all $c\in \F_3, \beta|_{\ell = c}$
  is $\Delta/27$ far from $\Pe^{n-1}_{2n-2d-1}$. By applying a change of
  basis, we can assume that $\ell = x_n$.

  Let $\beta = (x^2_n - 1) \gamma+ x_n \eta+ \theta$ and $q= (x_n^2-1)
  r + (x_n-1) s+ t$ where $\gamma,\eta,\theta,r,s,t$ do not depend
  on $x_n$. Note that $\theta - \gamma, \theta +\eta, \theta -\eta$
  are $D/27$ far from $\Pe^{n-1}_{2n-2d-1}$. Expanding the product $\beta q$, we have
$$\beta q = (x_n^2 -1)\left((\theta  -\gamma) r+\gamma t+\eta s -\gamma
s\right)+(x_n-1)\left((\theta  -\eta) s+\eta t\right)+(\theta +\eta) t.$$
Comparing terms, we observe that  $\beta q\in \Pe^n_{2n-2d -1 +k}$ iff the following are true:
\begin{enumerate}
\item $(\theta  -\gamma) r+\gamma t+\eta s -\gamma s \in \Pe^{n-1}_{2n-2d -1 +k -2}$
\item $(\theta  - \eta) s + \eta t \in \Pe^{n-1}_{2n-2d -1 +k -1}$
\item $(\theta  +\eta) t \in \Pe^{n-1}_{2n-2d -1 +k}$
\end{enumerate}
Since $r \in \Pe^n_{k-2}, s \in \Pe^n_{k-1}, t \in \Pe^n_k$, this is
equivalent to the following (written in reverse order):
\begin{enumerate}
\item $t \in B_{d-1,k}^{n-1}(\theta +\eta)$ 
\item $s \in -\eta t + B_{d-1,k-1}^{n-1}(\theta -\eta)$
\item $r \in \gamma s -\eta s -\gamma t + B^{n-1}_{d-1,k-2}(\theta-\gamma)$
\end{enumerate}
Since $t,s,r$ belongs to sets with the same size as
$B^{n-1}_{d-1,k}(\theta +\eta), B^{n-1}_{d-1,k-1}(\theta -\eta),
B^{n-1}_{d-1,k-2}(\theta-\gamma)$ respectively and each choice gives a
distinct element of $B^n_{d,k}(\beta)$, we get the following equality.
$$\dim(B^n_{d,k}(\beta)) = \dim(B^{n-1}_{d-1,k}(\theta +\eta)) + \dim(B^{n-1}_{d-1,k-1}(\theta -\eta)) + \dim(B^{n-1}_{d-1,k-2}(\theta-\gamma))$$
Combining this with $\dim(\Pe^n_k) = \dim(\Pe^{n-1}_k) + \dim(\Pe^{n-1}_{k-1}) +
\dim(\Pe^{n-1}_{k-2})$, we obtain
\begin{align*}
\dim(\Pe^n_k/B^n_{d,k}(\beta)) &=
\dim(\Pe^{n-1}_k/B^{n-1}_{d-1,k}(\theta +\eta)) +
\dim(\Pe^{n-1}_{k-1}/B^{n-1}_{d-1,k-1}(\theta -\eta)) +
\dim(\Pe^{n-1}_{k-2}/B^{n-1}_{d-1,k-2}(\theta-\gamma))\\
&\geq  \Phi_{d-1,k}(D/27) + \Phi_{d-1,k-1}(D/27) + \Phi_{d-1,k-2}(D/27).
\end{align*}
The last inequality follows from the fact that $\theta - \gamma, \theta +\eta, \theta -\eta$
  are $D/27$ far from $\Pe^{n-1}_{2n-2d-1} = \Pe^{n-1}_{2(n-1)-2(d-1)-1}$. Thus, proved.
\end{proof}

To prove \lref[Lemma]{lem:far-rank}, we start with
$\Phi_{d,d}(3^{d/2})$ and apply \lref[Claim]{claim:rec} recursively
$d/6 -2$ times and finally use the base case from \lref[Claim]{claim:base-case}
(this can be done as long as $d/6 - 2 \leq d/2$). This gives $\rank(Q^\beta) \geq \Phi_{d,d}(3^{d/2})\geq 3^{d/6-2}\geq 3^{d/9}$ as long as $d_{0}$ is large enough.

\end{proof}

\section{Hardness of coloring 2-colorable 8-uniform hypergraphs}\label{sec:2c8u}

We prove the theorem by a reduction from $\TSAT$ via the instances of
the repeated label cover problem obtained in
\lref[Theorem]{thm:label-cover}. Let $r\in\naturals$ be a parameter
that we will fix later and let $I(\phi)$ be an instance of
$r$-repeated label cover obtained in \lref[Theorem]{thm:label-cover}
starting from a $\TSAT$ instance $\phi$.

We denote by $G = (U,V,E)$ the underlying left and right regular
bipartite graph. For $u\in U$ and $i\in [3r]$, fix functions
$f_i^u:\{0,1\}^{3r}\rightarrow \{0,1\}$ as in
\lref[Definition]{def:label-cover}. Throughout this section, we work
over $\F_2$. For a degree parameter $d$ that we will determine later
and a vertex $u\in U$, we define the subspace $J_u$ of $\Pe^{3r}_{d}$
as follows:
\[
J_u := \left\{ \sum_{i=1}^{3r} r_i f_i^u : r_i \in \Pe^{3r}_{(d-3)} \right\}.
\]
Note that since each $f_i^u$ depends only on $3$ variables, it is a
polynomial of degree at most $3$ and hence, $J_u$ is indeed a subspace
of $\Pe^{3r}_d$. Let $N_u$ denote the cardinality of the quotient
space $\Pe^{3r}_d/J_u$.

We now define the hypergraph $H$ produced by the reduction. The
vertices of $H$ --- denoted $V(H)$ --- are obtained by replacing each
$u\in U$ by a block $\mc{B}_u$ of $N_u$ vertices, which we identify
with elements of $\Pe^{3r}_d/J_u$. Let $N$ denote $|V(H)| = \sum_{u\in
  U} N_u$.

We think of a $2$-coloring of $V(H)$ as a map from $V(H)$ to
$\F_2$. Given a coloring $A:V(H)\rightarrow \field_2$, we denote by
$A_u:\Pe^{3r}_d/J_u\rightarrow \F_2$ the restriction of $A$ to the
block $\mc{B}_u$ (under our identification of $\mc{B}_u$ with
$\Pe^{3r}_d/J_u$). Let $A_u':\Pe^{3r}_d\rightarrow \field_2$ denote
the lift of $A_u$ as defined in \lref[Fact]{fact:ideallift}.

The (weighted) edge set $E(H)$ of $H$ is specified implicitly by the following PCP verifier for the label cover instance $I(\phi)$, which expects as its input a $2$-coloring $A:V(H)\rightarrow \field_2$.\\

\noindent
{\bfseries $2$-Color $8$-Uniform Test}$(d)$
\begin{enumerate}
\item Choose a uniformly random $v\in V$ and then choose $u,w\in U$
  uniformly random neighbors of $v$ (by the right regularity of $G$,
  both $(u,v)$ and $(u,w)$ are unifom random edges in $E$). Let $\pi$
  denote $\pi_{uv}:\F_2^{3r}\rightarrow \F_2^r$ and similarly, let
  $\pi'$ be $\pi_{wv}$.
\item Choose $f \in \Pe^{r}_d$, $e_1,e_2,e_3,e_4 \in \Pe^{3r}_d$, and
  $g_1,g_2\in \Pe^{3r}_{d/4}$ and $h_1,h_2,h_3,h_4 \in
  \Pe^{3r}_{3d/4}$ independently and uniformly at random. Define
  functions $\eta_1,\eta_2,\eta_3,\eta_4 \in \Pe^{3r}_d$ as follows.
\begin{align*}
\eta_1 &:=  1+ f\circ \pi + g_1h_1,&\eta_3&:= f\circ \pi'  + g_2h_3,\\
\eta_2&:= 1+f\circ \pi + (1+g_1)h_2,&\eta_4&:=f\circ \pi' + (1+g_2)h_4.
\end{align*}
\item Accept if and only if $A_u'(e_1), A_u'(e_1 +\eta_1), A_u'(e_2),
  A_u'(e_2+\eta_2), A_w'(e_3),  A_w'(e_3+\eta_3), A_w'(e_4), A_w'(e_4+\eta_4)$ are not all equal.
\end{enumerate}

We now analyze the above test.

\begin{lemma}[Completeness] 
\label{lem:2c8u-comp}
If $\phi$ is satisfiable, then there exists a $2$-coloring
$A:V(H)\rightarrow \field_2$ such that the verifier accepts with
probability $1$. In other words, the hypergraph $H$ is $2$-colorable.
\end{lemma}

\begin{proof}
  Since $\phi$ is satisfiable, \lref[Theorem]{thm:label-cover} tells
  us that there are labelings $L_U:U\rightarrow \F_2^{3r}$ and
  $L_V:V\rightarrow \F_2^r$ such that for all $u\in U$, $L_U(u)$
  satisfies all the clauses in $U$ and moreover, for every edge
  $(u,v)\in E$, we have $\pi_{uv}(L_U(u)) = L_V(v)$. Fix such
  $L_U,L_V$. Let $a_u$ denote $L_U(u)$ for any $u\in U$ and $b_v$
  denote $L_V(v)$ for any $v\in V$.

  Now, the coloring $A:V(H)\rightarrow \F_2$ is defined to ensure that
  for each $u\in U$, its restriction $A_u$ is such that its lift $A_u'
  = \LC_d(a_u)$. Note that this makes sense since $\LC_d(a_u)$ is
  folded over $J_u$: indeed, given any $g\in\Pe^{3r}_{d}$ and $h =
  \sum_{i}r_if^u_i\in J_u$, we have $\LC_d(a_u)(g+h) = g(a_u)+h(a_u) =
  g(a_u)$ as $h(a_u)=\sum_ir_i(a_u) f^u_i(a_u) = 0$ for any satisfying
  assignment $a_u$ of the clauses corresponding to $u$.


  We now show that the verifier accepts $A$ with probability $1$. Fix
  any choices of $v\in V$ and $u,w\in U$, $f$, $e_i,h_i$ ($i\in [4]$)
  and $g_i$ ($i\in [2]$) as in the test. By the definitions of $L_U$
  and $L_V$, we must have $\pi(a_u) = \pi'(a_w) = b_v$. This implies
  that the $8$ positions in $A$ viewed by the verifier respectively
  contain the following values:
%
\begin{align*}
  e_1(a_u), \ &e_1(a_u)+1+f(b_v) + g_1(a_u)h_1(a_u),\\
  e_2(a_u), \ &e_2(a_u) +1+f(b_v) + (1+g_1(a_u))h_2(a_u),\\
 e_3(a_w), \ &e_3(a_w)+f(b_v) + g_2(a_w)h_3(a_w),\\
 e_4(a_w), \  &e_4(a_w)+f(b_v)
 + (1+g_2(a_w))h_4(a_w).
\end{align*}

If $f(b_v)=0$, then either the first two values or the third and fourth values are unequal, whereas if $f(b_v)=1$, then one of the last two pairs must be unequal. Thus, the verifier always accepts.
\end{proof}

\begin{remark}
\label{rem:doubling}
\lref[Lemma]{lem:2c8u-comp} actually yields a stronger statement. Let us group the probes of the verifier as $(e_i,e_i+\eta_i)$ for $i\in [4]$. Then, for the given coloring $A$ in \lref[Lemma]{lem:2c8u-comp} and any random choices of the verifier, there is some $i\in [4]$ such that $A$ is not constant on inputs in the $i$th group. We use this in \lref[Section]{sec:4u4c} to devise a $4$-query verifier over an alphabet of size $4$.
\end{remark}


\begin{lemma}[Soundness]
\label{lem:2c8u-sound}
Let $d\geq 8$ be a multiple of $4$ , $\delta>0$ and $\epsilon_0$ be
the constant from \lref[Theorem]{thm:label-cover}. If $\phi$ is
unsatisfiable and $H$ contains an independent set of size $\delta N$,
then $\delta^8\leq 2^{d/2}\cdot 2^{-\epsilon_0 r} + 2^{-4\cdot
  2^{-d/4}}$.
\end{lemma}

\begin{proof}
Fix any independent set $\mc{I}\subseteq V(H)$ of size $\delta N$. Let $A:V(H)\rightarrow \{0,1\}$ be the indicator function of $\mc{I}$. For $u\in U$, let $A_u:\Pe^{3r}_d/J_u\rightarrow \{0,1\}$ denote the restriction of $A$ to the block of vertices corresponding to $u$ and let $A_u':\Pe^{3r}_d\rightarrow \{0,1\}$ be the lift of $A_u$. Note that we have $\E_{(g + J_u)\in \Pe^{3r}_d/J_u}[A_u(g)] = \E_{g\in \Pe^{3r}_d}[A_u'(g)]$ for any $u\in U$. In particular,
\begin{equation}
\label{eq:2c8u-1}
\E_{u\in U} \E_{g\in \Pe^{3r}_d}\left[A_u'(g)\right] = \E_{u\in U} \E_{(g+J_u)\in \Pe^{3r}_d/J_u}\left[A_u(g)\right]\geq \delta.
\end{equation}

Since $\mc{I}$ is an independent set, in particular it must be the case that the probability that a random edge (chosen according to the probability distribution defined on $E(H)$ by the PCP verifier) completely lies inside $\mc{I}$ is $0$. We note that another expression for this probability is given by the quantity $\E_{v\in V, u,w\in U}[Q(v,u,w)]$ where $v\in V$ and $u,w\in U$ are as chosen by the PCP verifier described above and $Q(v,u,w)$ is defined as follows:
\begin{equation*}
\label{eq:2c8u-2}
Q(v,u,w) := \avg{\substack{\substack{\eta_1,\eta_2\\ \eta_3,\eta_4}}}{\avg{\substack{e_1,e_2\\ e_3,e_4}} {\prod_{i\in [2]}A_u'(e_i)A_u'(e_i+\eta_i)A_w'(e_{i+2})A_w'(e_{i+2}+\eta_{i+2})}}.
\end{equation*}

We analyze the right hand side of the above using its Fourier expansion (see \lref[Lemma]{lem:fourier}). 
As defined in \lref[Section]{sec:ldlc}, let $\Lambda^{3r}_d$ be a set
of minimum weight coset representatives of the cosets of $(\Pe^{3r}_d)^\perp$ in $\mathfrak F_{3r}$. Standard computations yield the following:
%
\begin{align}
Q(v,u,w) 
&= \sum_{\substack{\alpha_1,\alpha_2\\\beta_1,\beta_2\in \Lambda^{3r}_d}} \underbrace{ \left(\prod_{i\in [2]}  \widehat{A_u'}(\alpha_i)^2 \widehat{A_w'}(\beta_i)^2\right) \avg{\substack{\eta_1,\eta_2\\ \eta_3,\eta_4}}{\prod_{i\in [2]}\chi_{\alpha_i}(\eta_i)\chi_{\beta_i}(\eta_{i+2})}}_{\xi_{v,u,w}(\alpha_1,\alpha_2,\beta_1,\beta_2)}.\label{eq:2c8u-3}
\end{align}

When $v,u,w$ are clear from context, we use $\xi(\alpha_1,\alpha_2,\beta_1,\beta_2)$ instead of $\xi_{v,u,w}(\alpha_1,\alpha_2,\beta_1,\beta_2)$.

We analyze the above expression by breaking it up as follows. Let
$$\Far := \{(\alpha_1,\alpha_2,\beta_1,\beta_2)\in (\Lambda^{3r}_d)^4
: \max\{\Delta(\alpha_i,\Pe^{3r}_d),\Delta(\beta_i,\Pe^{3r}_d)\} \geq
2^{d/2}\}, \ \Near := (\Lambda^{3r}_d)^4\setminus \Far.$$ 
We now make the following claim for every $v,u,w$, the proof of which is deferred to the
end of the section.

\begin{claim}
\label{clm:2c8u-far}
$\sum_{(\alpha_1,\alpha_2,\beta_1,\beta_2)\in
  \Far}|\xi(\alpha_1,\alpha_2,\beta_1,\beta_2)|\leq 2^{-4\cdot 2^{d/4}}$.
\end{claim}

Substituting in \eqref{eq:2c8u-3}, we have for any $v\in V$ and $u,w\in U$,
\begin{align}
Q(v,u,w) &\geq \sum_{(\alpha_1,\alpha_2,\beta_1,\beta_2)\in \Near} \xi(\alpha_1,\alpha_2,\beta_1,\beta_2) - \sum_{(\alpha_1,\alpha_2,\beta_1,\beta_2)\in \Far} \left|\xi(\alpha_1,\alpha_2,\beta_1,\beta_2)\right|\notag\\
&\geq \sum_{(\alpha_1,\alpha_2,\beta_1,\beta_2)\in \Near} \xi(\alpha_1,\alpha_2,\beta_1,\beta_2) - 2^{-4\cdot 2^{-d/4}}.\label{eq:2c8u-4}
\end{align}

Now fix any $(\alpha_1,\alpha_2,\beta_1,\beta_2)\in \Near$. We analyze the expectation term in $\xi(\alpha_1,\alpha_2,\beta_1,\beta_2)$ further as follows.
\begin{align}
&\avg{\substack{\eta_1,\eta_2\\ \eta_3,\eta_4}}{\prod_{i\in [2]}\chi_{\alpha_i}(\eta_i)\chi_{\beta_i}(\eta_{i+2})}\notag\\
=& \avg{\substack{g_1,g_2,f\\ h_1,\ldots,h_4}}{
\chi_{\alpha_1}(1+f\circ \pi + g_1h_1)
\chi_{\alpha_2}(1+f\circ \pi + (1+g_1)h_2)
\chi_{\beta_1}(f\circ \pi' + g_2h_3)
\chi_{\beta_2}(f\circ \pi' + (1+g_2)h_4)
}\notag\\
=& \avg{\substack{g_i,h_j}}{\prod_{i\in [2]}\chi_{\alpha_i}(1+(1+i+g_1)h_i)\chi_{\beta_i}((1+i+g_2)h_{i+2})\cdot\avg{f}{\chi_{\pi_2(\alpha_1+\alpha_2) + \pi'_2(\beta_1+\beta_2)}(f)}\ }.\label{eq:2c8u-5}
\end{align}
where $\pi_2$ and $\pi'_2$ are as defined in \lref[Lemma]{lem:char-projection}. The innermost expectation is $0$ unless $\chi_{\pi_2(\alpha_1+\alpha_2) + \pi'_2(\beta_1+\beta_2)}$ is the trivial character on $\Pe^r_d$ or equivalently, $\gamma := \pi_2(\alpha_1+\alpha_2) + \pi'_2(\beta_1+\beta_2)\in (\Pe^{r}_{d})^\perp$.

We claim that this implies that $\gamma = 0$. To see this, we observe from the definition of $\pi_2$ and $\pi'_2$ that $|\supp(\gamma)|\leq \sum_{i\in [2]}|\supp(\alpha_i)| + |\supp(\beta_i)| \leq 4\cdot 2^{d/2}$, since $(\alpha_1,\alpha_2,\beta_1,\beta_2)\in\Near$ and $|\supp(\alpha)| = \Delta(\alpha,(\Pe^{3r}_{d})^\perp)$ for $\alpha\in \Lambda^{3r}_d$. However, if $\gamma\neq 0$ and $\gamma\in (\Pe^{r}_d)^\perp$, by \lref[Lemma]{lem:SZ}, we must have $|\supp(\gamma)|\geq 2^{d} > 4\cdot 2^{d/2}$ since $d\geq 8$. This implies that $\gamma = 0$. Substituting in \eqref{eq:2c8u-5}, we get
\begin{align}
\avg{\substack{\eta_1,\eta_2\\ \eta_3,\eta_4}}{\prod_{i\in [2]}\chi_{\alpha_i}(\eta_i)\chi_{\beta_i}(\eta_{i+2})}
=& \left\{
\begin{array}{l}
\text{$0$, if $\pi_2(\alpha_1+\alpha_2) + \pi'_2(\beta_1+\beta_2)\neq 0$},\\
\text{$\avg{\substack{g_j,h_i}}{\prod_{i\in [2]}\chi_{\alpha_i}(1+(1+i+g_1)h_i)\chi_{\beta_i}((1+i+g_2)h_{i+2})}$, otherwise.}
\end{array}
\right.\label{eq:2c8u-6}
\end{align}

Substituting back in \eqref{eq:2c8u-4}, we have
\begin{align}
Q(v,u,w) &= \sum_{\substack{(\alpha_1,\alpha_2,\beta_1,\beta_2)\in \Near:\\ \pi_2(\alpha_1+\alpha_2) + \pi'_2(\beta_1+\beta_2) = 0}} \xi(\alpha_1,\alpha_2,\beta_1,\beta_2) - 2^{-4\cdot 2^{-d/4}}.\label{eq:2c8u-7}
\end{align}

We partition the terms in the above sum further into $\Near_0 := \{(\alpha_1,\alpha_2,\beta_1,\beta_2)\in \Near: \pi_2(\alpha_1+\alpha_2) = \pi'_2(\beta_1+\beta_2) = 0\}$ and $\Near_1 := \{(\alpha_1,\alpha_2,\beta_1,\beta_2)\in \Near : \pi_2(\alpha_1+\alpha_2) = \pi'_2(\beta_1+\beta_2) \neq 0\}$.


\begin{claim}
\label{clm:2c8u-near-1}
$\avg{v,u,w}{\sum_{(\alpha_1,\alpha_2,\beta_1,\beta_2)\in \Near_1}|\xi_{v,u,w}(\alpha_1,\alpha_2,\beta_1,\beta_2)|}\leq 2^{d/2} \cdot 2^{-\epsilon_0 r}$.
\end{claim}

\begin{claim}
\label{clm:2c8u-near-0}
$\avg{v,u,w}{\sum_{(\alpha_1,\alpha_2,\beta_1,\beta_2)\in \Near_0}\xi_{v,u,w}(\alpha_1,\alpha_2,\beta_1,\beta_2)}\geq \delta^8$.
\end{claim}

Assuming these claims for now, we can finish the proof of \lref[Lemma]{lem:2c8u-sound} as follows. By \eqref{eq:2c8u-7}, 
\begin{align*}
0 &= \avg{v,u,w}{Q(v,u,w)} \\
&\geq \avg{v,u,w}{\sum_{(\alpha_1,\alpha_2,\beta_1,\beta_2)\in \Near_0} \xi_{v,u,w}(\alpha_1,\alpha_2,\beta_1,\beta_2)} - \avg{v,u,w}{\sum_{(\alpha_1,\alpha_2,\beta_1,\beta_2)\in \Near_1} |\xi_{v,u,w}(\alpha_1,\alpha_2,\beta_1,\beta_2)|} - 2^{-4\cdot 2^{-d/4}}\\
&\geq \delta^8 - 2^{d/2}\cdot 2^{-\epsilon_0 r}  - 2^{-4\cdot 2^{-d/4}}.\qedhere
\end{align*}
\end{proof}

We now fill in the proofs of
Claims~\ref{clm:2c8u-far}--\ref{clm:2c8u-near-0}.

\begin{proof}[Proof of {\lref[Claim]{clm:2c8u-far}}]
Fix any $(\alpha_1,\alpha_2,\beta_1,\beta_2)\in \Far$. Conditioned on any choice of $f$, the expectation term in $|\xi(\alpha_1,\alpha_2,\beta_1,\beta_2)|$ may be bounded as follows:
\begin{align}
&\left|\avg{\substack{\eta_1,\eta_2\\ \eta_3,\eta_4}}{\prod_{i\in [2]}\chi_{\alpha_i}(\eta_i)\chi_{\beta_i}(\eta_{i+2})}\right|\notag\\ 
&= \left|\avg{\substack{g_1,g_2\\ h_1,\ldots,h_4}}{
\chi_{\alpha_1}(1+f\circ \pi + g_1h_1)
\chi_{\alpha_2}(1+f\circ \pi + (1+g_1)h_2)
\chi_{\beta_1}(f\circ \pi' + g_2h_3)
\chi_{\beta_2}(f\circ \pi' + (1+g_2)h_4)
}\right|\notag\\
&\leq \avg{g_1,g_2}{\prod_{i\in [2]}
\left|\avg{h_i}{\chi_{\alpha_i}(1+f\circ \pi + (1+i+g_1)h_i)}\right|\cdot
\left|\avg{h_{i+2}}{\chi_{\beta_i}(f\circ \pi' + (1+i+g_2)h_{i+2})}\right|}\notag\\
&= \avg{g_1,g_2}{\prod_{i\in [2]}
\left|\avg{h_i}{\chi_{\alpha_i}((1+i+g_1)h_i)}\right|\cdot
\left|\avg{h_{i+2}}{\chi_{\beta_i}((1+i+g_2)h_{i+2})}\right|}\notag\\
&\leq \avg{g_1,g_2}{\min\left\{
\left|\avg{h_i}{\chi_{\alpha_i}((1+i+g_1)h_i)}\right|, 
\left|\avg{h_{i+2}}{\chi_{\beta_i}((1+i+g_2)h_{i+2})}\right|: i\in [2]\right\}}\notag\\
&\leq \min\left\{ \avg{g_1}{\left|\avg{h_i}{\chi_{\alpha_i}((1+i+g_1)h_i)}\right|}, 
\avg{g_2}{\left|\avg{h_{i+2}}{\chi_{\beta_i}((1+i+g_2)h_{i+2})}\right|} : i\in [2] \right\}.
\label{eq:2c8u-8}
\end{align}

Note that for any $i\in [2]$, $(1+i+g_1)$ and $(1+i+g_2)$ are uniformly random elements of $\Pe^{3r}_{d/4}$ that are independent of $h_1,\ldots,h_4$. Moreover, since $(\alpha_1,\alpha_2,\beta_1,\beta_2)\in \Far$, we know that there is a $\gamma\in \{\alpha_1,\alpha_2,\beta_1,\beta_2\}$ such that $\Delta(\gamma, (\Pe^{3r}_{d})^\perp)\geq 2^{d/2}$. Therefore, by \lref[Theorem]{thm:DG}, we have
\[
\avg{g\in \Pe^{3r}_{d/4}}{\left|\avg{h\in
      \Pe^{3r}_{3d/4}}{\chi_{\gamma}(gh)}\right|}\leq 2^{-4\cdot 2^{d/4}}.
\]
Substituting the above in \eqref{eq:2c8u-8}, we obtain
\[
\left|\avg{\substack{\eta_1,\eta_2\\ \eta_3,\eta_4}}{\prod_{i\in
      [2]}\chi_{\alpha_i}(\eta_i)\chi_{\beta_i}(\eta_{i+2})}\right|
\leq 2^{-4\cdot 2^{d/2}}.
\]
Thus, we obtain
\begin{align*}
\sum_{(\alpha_1,\alpha_2,\beta_1,\beta_2)\in \Far}|\xi(\alpha_1,\alpha_2,\beta_1,\beta_2)| &\leq 
 2^{-4\cdot 2^{d/2}}\cdot \sum_{\alpha_1,\alpha_2,\beta_1,\beta_2\in \Lambda^{3r}_d} \left(\prod_{i\in [2]}\widehat{A_u'}(\alpha_i)^2 \widehat{A_w'}(\beta_i)^2\right)
\leq 2^{-4\cdot 2^{d/2}},
\end{align*}
where the last inequality follows from Parseval's identity and the fact that $|A(x)|\leq 1$ for all $x\in V(H)$.
\end{proof}

\begin{proof}[Proof of {\lref[Claim]{clm:2c8u-near-1}}]
We use a Fourier decoding argument. Formally, we sample random labelings $L_U:U\rightarrow \F_2^{3r}$ and $L_V:V\rightarrow \F_3^r$ such that
\begin{align}
\prob{(u,v)\in E, L_U,L_V}{\pi_{uv}(L_U(u)) = L_V(v)}\geq \frac{1}{2^{d/2}} \avg{v,u,w}{\sum_{(\alpha_1,\alpha_2,\beta_1,\beta_2)\in \Near_1}|\xi_{v,u,w}(\alpha_1,\alpha_2,\beta_1,\beta_2)|}.\label{eq:2c8u-9}
\end{align}
Since $OPT(I(\phi))\leq 2^{-\epsilon_0 r}$, the left hand side of the above inequality is at most $2^{-\epsilon_0 r}$. This implies the claim.

Define $L_U:U\rightarrow \F_2^{3r}$ as follows: given $u\in U$, we sample a random pair $\alpha_1,\alpha_2\in\Lambda^{3r}_d$ such that $|\alpha_1|,|\alpha_2| < 2^{d/2}$ with probability proportional to $\widehat{A'_u}(\alpha_1)^2\widehat{A'_u}(\alpha_2)^2$ and set $L_U(u)$ to be $a_u$ for a uniformly random $a_u$ chosen from $\supp(\alpha_1)\cup \supp(\alpha_2)$. Since $|\alpha_1|,|\alpha_2|<2^{d/2} < 2^{d-4}$, by \lref[Lemma]{lem:goodsupport}, any $\alpha_1,\alpha_2$ sampled as above is supported only on satisfying assignments of all the clauses in $u$. 

We also define $L_V:V\rightarrow \F_2^{r}$ similarly: given $v\in V$, we sample a random neighbor $w\in U$ of $v$ and choose at random a pair $\beta_1,\beta_2\in\Lambda^{3r}_d$ such that $|\beta_1|,|\beta_2|<  2^{d/2}$ with probability proportional to $\widehat{A'_w}(\beta_1)^2\widehat{A'_w}(\beta_2)^2$ and set $L_V(v)$ to be $\pi_{wv}(a_w)$ for a uniformly random $a_w$ chosen from $\supp(\beta_1)\cup \supp(\beta_2)$.

Let $(u,v)\in E$ be a uniformly random edge of $G$ and consider the probability that $\pi_{uv}(L_U(u)) = L_V(v)$. This probability can clearly be lower bounded as follows.
\begin{align*}
\prob{(u,v)\in E, L_U,L_V}{\pi(L_U(u)) = L_V(v)} &\geq \avg{v,u,w}{\sum_{\substack{(\alpha_1,\alpha_2,\beta_1,\beta_2)\in \Near:\\ \pi(\supp(\alpha_1)\cup \supp(\alpha_2))\cap \\ \pi'(\supp(\beta_1)\cup \supp(\beta_2))\neq \emptyset}} \prod_{i\in [2]}\widehat{A_u'}(\alpha_i)^2\widehat{A_w'}(\beta_i)^2}\cdot \frac{1}{2^{d/2}},
\end{align*}
where $\pi$ denotes $\pi_{uv}$ and $\pi'$ denotes $\pi_{wv}$. Observe
that if $(\alpha_1,\alpha_2,\beta_1,\beta_2)\in \Near_1$, then
$\pi_2(\alpha_1 + \alpha_2) = \pi'_2(\beta_1+\beta_2) \neq 0$ and in
particular, $\pi(\supp(\alpha_1)\cup \supp(\alpha_2))\cap
\pi'(\supp(\beta_1)\cup \supp(\beta_2))\neq \emptyset$. Therefore, we
get the following which implies \eqref{eq:2c8u-9} and hence proves the claim.
\begin{align*}
\prob{(u,v)\in E, L_U,L_V}{\pi(L_U(u)) = L_V(v)}
&\geq \frac{1}{2^{d/2}}\avg{v,u,w}{\sum_{\substack{(\alpha_1,\alpha_2,\beta_1,\beta_2)\in \Near_1}} \prod_{i\in [2]}\widehat{A_u'}(\alpha_i)^2\widehat{A_w'}(\beta_i)^2}.\qedhere
\end{align*}
\end{proof}

\begin{proof}[Proof of {\lref[Claim]{clm:2c8u-near-0}}]
We argue below that for any $v\in V$ and its neighbours $u,w\in U$ and any $(\alpha_1,\alpha_2,\beta_1,\beta_2)\in \Near_0$, 
\begin{align}
\xi(\alpha_1,\alpha_2,\beta_1,\beta_2) \geq 0.\label{eq:2c8u-10}
\end{align}

Given \eqref{eq:2c8u-10}, we have
\begin{align*}
\avg{v,u,w}{\sum_{(\alpha_1,\alpha_2,\beta_1,\beta_2)\in \Near_0}\xi_{v,u,w}(\alpha_1,\alpha_2,\beta_1,\beta_2)} &\geq \avg{v,u,w}{\xi_{v,u,w}(0,0,0,0)} = \avg{v,u,w}{\widehat{A_u'}(0)^4 \widehat{A_w'}(0)^4}.
\end{align*}
Conditioned on $v\in V$, $u$ and $w$ are independent and randomly chosen neighbours of $v$. Thus, the above may be further lower bounded as follows. 
\begin{align*}
\avg{v,u,w}{\widehat{A_u'}(0)^4 \widehat{A_w'}(0)^4} &= \avg{v}{\left(\avg{u: (u,v)\in E}{\widehat{A_u'}(0)^4}\right)^2}
\geq \left(\avg{(u,v)\in E}{\widehat{A_u'}(0)}\right)^8 = \left(\avg{u\in U, g\in \Pe_d^{3r}}{{A_u'}(g)}\right)^8\geq
\delta^8,
\end{align*}
where the first inequality follows from repeated applications of the Cauchy-Schwarz inequality and the last from \eqref{eq:2c8u-1}.

For any $v,u,w$ and
$(\alpha_1,\alpha_2,\beta_1,\beta_2)\in \Near_0$, it remains to prove
\eqref{eq:2c8u-10} (i.e., non-negativity of
$\xi_{v,u,w}(\alpha_1,\alpha_2,\beta_1,\beta_2)$).
From
\eqref{eq:2c8u-3},  it suffices to argue the non-negativity of  
\begin{align}
\avg{\substack{\eta_1,\eta_2\\ \eta_3,\eta_4}}{\prod_{i\in [2]}\chi_{\alpha_i}(\eta_i)\chi_{\beta_i}(\eta_{i+2})} &= \avg{g_1,g_2}{\prod_{i\in [2]}\avg{h_i}{\chi_{\alpha_i}(1+(1+i+g_1)h_i)}\avg{h_{i+2}}{\chi_{\beta_i}((1+i+g_2)h_{i+2})}}\notag\\
&= \avg{g_1,g_2}{(-1)^{\sum_{x}\alpha_1(x) + \alpha_2(x)}\cdot \prod_{i\in [2]}\avg{h_i}{\chi_{\alpha_i(1+i+g_1)}(h_i)}\avg{h_{i+2}}{\chi_{\beta_i(1+i+g_2)}(h_{i+2})}},\label{eq:2c8u-11}
\end{align}
where we have used \eqref{eq:2c8u-6} for the first equality and the fact that $\chi_{\alpha}(gh) = \chi_{\alpha g}(h)$ for the second. We claim that all the terms inside the final expectation are non-negative. 

Firstly, since $(\alpha_1,\alpha_2,\beta_1,\beta_2)\in \Near_0$, we have $\pi_2(\alpha_1+\alpha_2) = 0$ and hence $(-1)^{\sum_x \alpha_1(x) + \alpha_2(x)}=(-1)^{\sum_y \pi_2(\alpha_1+\alpha_2)(y)} = 1$.  Secondly, the orthonormality of characters implies that for any $\alpha\in\mathfrak{F}_{3r}$, we have $\avg{h\in \Pe^r_{3d/4}}{\chi_{\alpha}(h)}\in \{0,1\}$ and hence non-negative. 

This shows that the right-hand side of \eqref{eq:2c8u-11} is non-negative. and hence proves \eqref{eq:2c8u-10}.
\end{proof}

\begin{proof}[Proof of {\lref[Theorem]{thm:2c8u}}]
  Given the completeness (\lref[Lemma]{lem:2c8u-comp}) and soundness
    (\lref[Lemma]{lem:2c8u-sound}), we only need to fix parameters.
Let $d = C\log r$ for a large enough constant $C\geq 8$ determined shortly. By \lref[Lemma]{lem:2c8u-sound},  if $H$ has an independent set of size $\delta N$, then $\delta^8 \leq 2^{d/2}\cdot 2^{-\epsilon_0 r} + 2^{-4\cdot 2^{-d/4}} < 2^{-\epsilon_0 r/2}$ for large enough $C>0$ and $r\in\naturals$. Hence, $H$ has no independent sets of $\delta'N$, where $\delta' = 2^{-\epsilon_0 r/16}$.

The hypergraph $H$ can be produced in time polynomial in $N = n^{O(r)}2^{r^{O(d)}} = n^{O(r)}2^{r^{O(\log r)}}$. Setting $r = 2^{\Theta(\sqrt{\log \log n})}$, we get $N = n^{2^{O(\sqrt{\log \log n})}}$, and $\delta' = 2^{-\Omega(r)} = 2^{-2^{\Theta(\sqrt{\log \log n})}}= 2^{-2^{\Theta(\sqrt{\log \log N})}}$, proving \lref[Theorem]{thm:2c8u}.
\end{proof}
%
%

\section{Hardness of coloring 4-colorable 4-uniform hypergraphs}\label{sec:4u4c}

This construction is motivated by \lref[Remark]{rem:doubling} above. We construct a new verifier each of whose queries correspond to two queries of the verifier described above. Let $I(\phi)$, $G = (U,V,E)$, and $J_u$ ($u\in U$) be as defined in \lref[Section]{sec:2c8u}.

Now the vertices of the hypergraph $H$ produced by the reduction denoted by $V(H)$ are obtained by replacing each $u\in U$ by a block $\mc{B}_u$ of $N_u^2$ vertices, which we identify with elements of $\Pe^{3r}_d/J_u\times \Pe^{3r}_d/J_u$. Let $N$ denote $|V(H)| = \sum_{u\in U} N_u^2$. 

We think of a $4$-coloring of $V(H)$ as a map from $V(H)$ to the $4$-element set $\F_2\times \F_2$. Given a coloring $A:V(H)\rightarrow \field_2\times \field_2$, we denote by $A_u:\Pe^{3r}_d/J_u \times \Pe^{3r}_d/J_u \rightarrow \F_2\times \F_2$ the restriction of $A$ to the block $\mc{B}_u$. Let $A_u':\Pe^{3r}_d\times\Pe^{3r}_d\rightarrow \field_2\times\field_2$ denote the lift of $A_u$ as defined by $A_u'(g_1,g_2) := A_u(g_1+J_u,g_2+J_u)$.

The verifier is defined as follows. The verifier is identical
to the verifier in \lref[Section]{sec:2c8u} but for the doubling of queries.\\

\noindent
{\bfseries $4$-Color $4$-Uniform Test}$(d)$
\begin{enumerate}
\item Choose a uniformly random $v\in V$ and then choose $u,w\in U$ uniformly random neighbors of $v$. Let $\pi$ denote $\pi_{uv}:\F_2^{3r}\rightarrow \F_2^r$ and similarly, let $\pi'$ be $\pi_{wv}$.
\item Choose $f \in \Pe^{r}_d$, $e_1,e_2,e_3,e_4 \in \Pe^{3r}_d$, and
  $g_1,g_2\in\Pe^{3r}_{d/4}$ and $h_1,h_2,h_3,h_4 \in \Pe^{3r}_{3d/4}$
  independently and uniformly at random. Define
  functions $\eta_1,\eta_2,\eta_3,\eta_4 \in \Pe^{3r}_d$ as follows.
\begin{align*}
\eta_1 &:=  1+ f\circ \pi + g_1h_1,&\eta_3&:= f\circ \pi'  + g_2h_3,\\
\eta_2&:= 1+f\circ \pi + (1+g_1)h_2,&\eta_4&:=f\circ \pi' + (1+g_2)h_4.
\end{align*}
\item Accept if and only if $A_u'(e_1,e_2),A_u'(e_1 +\eta_1,e_2+\eta_2),A_w'(e_3,e_4),A_w'(e_3+\eta_3,e_4+\eta_4)$ are not all equal.
\end{enumerate}

The analysis of the above test closely follows that of the $2$-color $8$-uniform test. 

\begin{lemma}[Completeness] 
\label{lem:4c4u-comp}
If $\phi$ is satisfiable, then there exists a $4$-coloring $A:V(H)\rightarrow \field_2\times \field_2$ such that the verifier accepts with probability $1$. In other words, the hypergraph $H$ is $4$-colorable.
\end{lemma}

\begin{proof}
Follows directly from \lref[Remark]{rem:doubling}.
\end{proof}

The soundness lemma requires us to perform Fourier analysis on functions $A:\Pe^{3r}_d\times \Pe^{3r}_d\rightarrow \{0,1\}$, for which we need the following easily verifiable facts.
\begin{fact}
\label{fact:prod-fourier}
Let $A:\Pe^{3r}_d\times \Pe^{3r}_d\rightarrow \C$ be any function. A non-zero function $\chi:\Pe^{3r}_d\times \Pe^{3r}_d\rightarrow\C$ is a character if $\chi(g_1+h_1,g_2+h_2) = \chi(g_1,g_2)\chi(h_1,h_2)$.
\begin{itemize}
\item $\chi:\Pe^{3r}_d\times \Pe^{3r}_d\rightarrow\C$ is a character
  if and only if there exist
  $(\alpha_1,\alpha_2)\in\mathfrak{F}_{3r}\times \mathfrak{F}_{3r}$
  such that $\chi(g_1,g_2) = \chi_{\alpha_1}(g_1)\chi_{\alpha_2}(g_2)$
  for any $g_1,g_2\in\Pe^{3r}_d\times \Pe^{3r}_d$ where
  $\chi_{\alpha_1}$ and $\chi_{\alpha_2}$ are characters of $\Pe^{3r}_d$.
\item $(\alpha_1,\alpha_2)$ and $(\beta_1,\beta_2)$ yield the same character if and only if $(\alpha_1-\alpha_2),(\beta_1-\beta_2)\in(\Pe^{3r}_d)^\perp$.
\item Folding: Fix $A:\Pe^{3r}_d\times \Pe^{3r}_d\rightarrow \C$ be
  any function folded over the subgroup $J\times J$ where $J :=
  \{\sum_{i=1}^k r_iq_i : r_i\in \Pe^{3r}_{d-3}\}$ and $q_1,\ldots,q_k
\in \Pe^{3r}_3$. Then, for any $(\alpha_1,\alpha_2)\in \mathfrak{F}_{3r}\times \mathfrak{F}_{3r}$ such that $|\alpha_j| := \Delta(\alpha_j,(\Pe^{3r}_d)^\perp) < 2^{d-3}$ for $j\in \{1,2\}$ and $\widehat{A}(\alpha_1,\alpha_2) \neq 0$, it must be the case that $\supp(\alpha_1)\cup\supp(\alpha_2)$ only contains $x$ such that $q_i(x) = 0$ for each $i\in [k]$.
\end{itemize}
\end{fact}

\begin{lemma}[Soundness]
\label{lem:4c4u-sound}
Let $d\geq 8$ be a multiple of $4$ , $\delta>0$ and $\epsilon_0$ be
the constant from \lref[Theorem]{thm:label-cover}.  If $\phi$ is unsatisfiable and $H$ contains an independent set of size $\delta N$, then $\delta^4\leq 2^{d/2}\cdot 2^{-\epsilon_0 r} + 2^{-4\cdot 2^{-d/4}}$.
\end{lemma}
The proof of \lref[Lemma]{lem:4c4u-sound} is similar to the proof of
\lref[Lemma]{lem:2c8u-sound}. 
The parameters are set
exactly as in \lref[Theorem]{thm:2c8u} to yield
\lref[Theorem]{thm:4u4c}.

\begin{proof}[Proof of {\lref[Lemma]{lem:4c4u-sound}}]
As the proof is similar to that of of \lref[Lemma]{lem:2c8u-sound}, we
only give a proof sketch, highlighting the salient differences.

As before, fix any independent set $\mc{I}\subseteq V(H)$ of size $\delta N$. Let $A:V(H)\rightarrow \{0,1\}$ be the indicator function of $\mc{I}$. We have $\E_{u\in U} \E_{g_1,g_2\in \Pe^{3r}_d}\left[A_u'(g_1,g_2)\right] \geq \delta$.

Again, we analyze $\E_{v\in V, u,w\in U}[Q(v,u,w)]$, which gives the probability that a random edge (chosen according to the probability distribution defined on $E(H)$ by the PCP verifier) completely lies inside the independent set $\mc{I}$, and is hence $0$. Here, $Q(v,u,w)$ is defined as follows:
\begin{equation*}
\label{eq:4c4u-2}
Q(v,u,w) := \avg{\substack{\substack{\eta_1,\eta_2\\ \eta_3,\eta_4}}}{\avg{\substack{e_1,e_2\\ e_3,e_4}} {A_u'(e_1,e_2)A_u'(e_1 +\eta_1,e_2+\eta_2)A_w'(e_3,e_4)A_w'(e_3+\eta_3,e_4+\eta_4)}}.
\end{equation*}

The Fourier expansion of this expression (see
\lref[Fact]{fact:prod-fourier}) yields the following. From
\lref[Fact]{fact:prod-fourier}, we have that 
$\mc{C}'_d := \Lambda^{3r}_d\times\Lambda^{3r}_d$ gives us all the distinct characters of $\Pe^{3r}_d\times \Pe^{3r}_d$. Standard computations give us
\begin{align*}
Q(v,u,w) &= \sum_{\substack{\alpha_1,\alpha_2\\\beta_1,\beta_2\in \Lambda^{3r}_d}} \underbrace{\widehat{A_u'}(\alpha_1,\alpha_2)^2 \widehat{A_w'}(\beta_1,\beta_2)^2 \avg{\substack{\eta_1,\eta_2\\ \eta_3,\eta_4}}{\prod_{i\in [2]}\chi_{\alpha_i}(\eta_i)\chi_{\beta_i}(\eta_{i+2})}}_{\xi_{v,u,w}'(\alpha_1,\alpha_2,\beta_1,\beta_2)}.\label{eq:4c4u-3}
\end{align*}
As in \lref[Lemma]{lem:2c8u-sound}, let $\Far :=
\{(\alpha_1,\alpha_2,\beta_1,\beta_2)\in (\Lambda^{3r}_d)^4 :
\max\{\Delta(\alpha_i,\Pe^{3r}_d),\Delta(\beta_i,\Pe^{3r}_d)\} \geq
2^{d/2}\}$, $\Near := (\Lambda^{3r}_d)^4\setminus \Far$, $\Near_0 :=
\{(\alpha_1,\alpha_2,\beta_1,\beta_2)\in \Near :
\pi_2(\alpha_1+\alpha_2) = \pi'_2(\beta_1+\beta_2) = 0\}$, and
$\Near_1 := \{(\alpha_1,\alpha_2,\beta_1,\beta_2)\in \Near :
\pi_2(\alpha_1+\alpha_2) = \pi'_2(\beta_1+\beta_2) \neq 0\}$. 

Note that the expectation term in $\xi'_{v,u,w}(\alpha_1,\alpha_2,\beta_1,\beta_2)$ is \emph{exactly} as that in $\xi_{v,u,w}(\alpha_1,\alpha_2,\beta_1,\beta_2)$ in \lref[Lemma]{lem:2c8u-sound}. This means that the remaining computations can be carried out almost exactly as in \lref[Lemma]{lem:2c8u-sound}. 

The following can be proved in the same way as Claims~\ref{clm:2c8u-far}, \ref{clm:2c8u-near-1} and \ref{clm:2c8u-near-0}. 

\begin{claim}
\label{clm:4c4u-far}
For any fixed $v,u,w$, we have $\sum_{(\alpha_1,\alpha_2,\beta_1,\beta_2)\in \Far}|\xi'_{v,u,w}(\alpha_1,\alpha_2,\beta_1,\beta_2)|\leq 2^{-4\cdot 2^{-d/4}}$.
\end{claim}

\begin{claim}
\label{clm:4c4u-near-1}
$\avg{v,u,w}{\sum_{(\alpha_1,\alpha_2,\beta_1,\beta_2)\in \Near_1}|\xi'_{v,u,w}(\alpha_1,\alpha_2,\beta_1,\beta_2)|}\leq 2^{d/2} \cdot 2^{-\epsilon_0 r}$.
\end{claim}

(There is a small difference here from the proof of \lref[Claim]{clm:2c8u-near-1} owing to the fact that the Fourier coefficients appearing in $\xi'_{v,u,w}(\alpha_1,\alpha_2,\beta_1,\beta_2)$ have a slightly different form. The only change that needs to be made is to sample $\alpha_1,\alpha_2\in\Lambda^{3r}_d$ and $\beta_1,\beta_2\in\Lambda^{3r}_d$ with probability proportional to $\widehat{A'}_u(\alpha_1,\alpha_2)^2$ and $\widehat{A'}_w(\beta_1,\beta_2)^2$ respectively.)

\begin{claim}
\label{clm:4c4u-near-0}
$\avg{v,u,w}{\sum_{(\alpha_1,\alpha_2,\beta_1,\beta_2)\in \Near_0}\xi'_{v,u,w}(\alpha_1,\alpha_2,\beta_1,\beta_2)}\geq \delta^4$.
\end{claim}


As in \lref[Lemma]{lem:2c8u-sound}, the above can be used to show:

\begin{align*}
0 &\geq  \avg{v,u,w}{\sum_{(\alpha_1,\alpha_2,\beta_1,\beta_2)\in \Near_0} \xi'_{v,u,w}(\alpha_1,\alpha_2,\beta_1,\beta_2) + \sum_{(\alpha_1,\alpha_2,\beta_1,\beta_2)\in \Near_1} \xi'_{v,u,w}(\alpha_1,\alpha_2,\beta_1,\beta_2)} - 2^{-4\cdot 2^{-d/4}}\\
&\geq \avg{v,u,w}{\sum_{(\alpha_1,\alpha_2,\beta_1,\beta_2)\in \Near_0} \xi'_{v,u,w}(\alpha_1,\alpha_2,\beta_1,\beta_2)} - \avg{v,u,w}{\sum_{(\alpha_1,\alpha_2,\beta_1,\beta_2)\in \Near_1} |\xi'_{v,u,w}(\alpha_1,\alpha_2,\beta_1,\beta_2)|} - 2^{-4\cdot 2^{-d/4}}\\
&\geq \delta^4 - 2^{d/2}\cdot 2^{-\epsilon_0 r}  - 2^{-4\cdot 2^{-d/4}}. 
\end{align*}
This completes the proof of \lref[Lemma]{lem:4c4u-sound}.
\end{proof}

\section{Hardness of coloring 3-colorable 3-uniform hypergraphs}\label{sec:3u3c}
This construction is an
adaptation of Khot's construction~\cite{Khot2002b} to the low-degree
long code setting. We prove the theorem by a reduction from $\TSAT$ via the instances of the multilayered label cover problem obtained in
\lref[Theorem]{thm:layered-label-cover}. Let $r,\ell,\eta$ be
parameters that will be determined later and let $I(\phi)$ be an
instance of the $r$-repeated $\ell$-layered $\eta$-smooth label cover instance with constraint graph
$G=(V_0,\dots,V_{\ell-1},\{E_{ij}\}_{0\leq i < j < \ell})$ obtained
from the $\TSAT$ instance $\phi$. We use the results from the
preliminaries with the field set to $\F_3=\{0,1,2\}$. For every layer $i$
and every vertex $v \in V_i$, let $\{c_1,\cdots c_{(T+\ell -i)r}\}$ be the
clauses corresponding to $v$ where $T = \lceil l/\eta\rceil$ as in \lref[Definition]{def:multilayer}. We construct polynomials
$\{p_1,\cdots p_{(T+\ell - i)r}\}$ of degree at most $6$ over $\F_3$
such that $p_j$ depends only on variables in $c_j$ with the following
properties. Let $a \in \F_3^3$. If $a \notin \{0,1\}^3$ then $p_j(a)
\neq 0$. Otherwise $p_j(a) = 0$ iff $c_j(a)=1$. For a degree parameter
$d$ that we will determine later, for each vertex $v$ define the
subspace $J_v$ as follows:
$$ J_v:= \left\{ \sum_i q_i p_i : q_i \in \Pe^{m_v}_{2d-6} \right\}
\text{ where } m_v := m_i=3(T+\ell-i)r +ir.$$

We now define the hypergraph $H$ produced by the reduction. The
vertices of $H$ --- denoted $V(H)$ --- are obtained by replacing each
$v\in G$ by a block $\mc{B}_v$ of
$N_v:=|\Pe^{m_v}_{2d}/J_v|$. vertices, which we identify with elements
of $\Pe^{m_v}_{2d}/J_v$. Let $N$ denote $|V(H)| = \sum_{v} N_v$.

We think of a $3$-coloring of $V(H)$ as a map from $V(H)$ to
$\F_3$. Given a coloring $A:V(H)\rightarrow \field_3$, we denote by
$A_v:\Pe^{m_v}_{2d}/J_v\rightarrow \F_3$ the restriction of $A$ to the
block $\mc{B}_v$. Let $A_v':\Pe^{m_v}_{2d}\rightarrow \field_3$ denote
the lift of $A_v$ as defined in \lref[Fact]{fact:ideallift}.

The (weighted) edge set $E(H)$ of $H$ is specified implicitly by the following PCP verifier.\\

\noindent
{\bfseries $3$-Color $3$-Uniform Test}$(d)$
\begin{enumerate}
\item Choose two layers $0\leq i < j < \ell$ uniformly at random and then choose a uniformly random edge $(u,v) \in E_{ij}$. Let $\pi$ denote $\pi_{uv}:\F_3^{m_u}\rightarrow\F_3^{m_v}$. 
\item Choose $p\in \Pe^{m_u}_d, g \in \Pe_{2d}^{m_u}$ and  $f\in  \Pe_{2d}^{m_v}$ independently and uniformly at random and let $g':=p^2+1 - g -f\circ \pi$.
\item Accept if and only if $A'_{v}(f),A'_{u}(g),A'_{u}(g')$ are not all equal.
\end{enumerate}

The above hypergraph construction explains the reasons (as in
~\cite{DinurRS2005,Khot2002b}) for using the multilayered label
cover. Unlike the constructions in the previous two sections, the
hyperedges in the 3-uniform case straddle both sides of the
corresponding edge $(u,v)$ in the label cover instance. Hence, if
constructed from the bipartite label cover, the corresponding
3-uniform hypergraph will also be bipartite and hence always
2-colorable irrespective of the label cover instance. Using the
multilayered construction gets around this problem.

\begin{lemma}[Completeness]
\label{lem:3c3u-completeness}
If $\phi \in \TSAT$, then there is proof $A:V(H)\rightarrow \F_3$
which the verifier accepts with probability $1$. In other words, the
hypergraph $H$ is $3$-colorable.
\end{lemma}
\begin{proof}
  Since $\phi \in \TSAT$, \lref[Theorem]{thm:layered-label-cover}
  tells us that there are labelings $L_i:V_i\rightarrow \{0,1\}^{m_i}$
  for $0\leq i <\ell$ which satisfy all the constraints in
  $I(\phi)$. For $\forall i,v\in V_i$, we set $A_v:\Pe^{m_v}_{2d}/J_v\rightarrow \F_3$ such that its
  lift $A'_v = \LC_{2d}(L_i(v))$. This is possible since $A'_v$ is
  folded over $J_v$. 
  For any edge $(u,v)$ between
  layers $i,j$, with labels $L_i(u) = a, L_j(v)=b$ such that
  $\pi(a)=b$, $(A'_{v}(f),A'_{u}(g),A'_{u}(g')) =
  (f(b),g(a),g'(a))$. The lemma follows by observing that
  $g'(a)+g(a)+f(b)\neq 0$ always (since $p^2(a) + 1 \neq 0$).
\end{proof}

\begin{lemma}[Soundness]
\label{lem:3c3u-soundness}
Let $\ell =32/\delta^2$. If $\phi \notin \TSAT$ and  $H$ contains a independent set of size $\delta|V(H)|$, then 
$$\delta^5/2^9 \leq 2^{-\Omega(r)}\cdot 3^d + \eta\cdot3^d +\exp(-3^{\Omega(d)}).$$
\end{lemma}
\begin{proof}
 Let $A:V(H)\rightarrow \{0,1\}$ be the characteristic function of the
 independent set of fractional size exactly $\delta$. We have that $\forall v, \E_{g \in P^{m_v}_{2d}/J_v}\left[A_v(g)\right] = \E_{g \in P^{m_v}_{2d}}\left[A'_v(g)\right]$ where $A'_v$ is the lift of $A_v$. Define
$$Q(u,v):=\E_{f,g,p}\left[A'_{v}(f)A'_{u}(g)A'_{u}(p^2+1 -f\circ \pi -g)\right].$$
Observe that $\E_{i,j,u,v}\left[Q(u,v)\right]=0$ as $A$ corresponds to an independent
set. Using \lref[Lemma]{lem:fourier}, we have the following Fourier
expansion of $Q$:
\begin{equation}
\label{eq:3u3c-quv}
 Q(u,v) = \sum_{\alpha,\beta,\gamma} \widehat A'_{v}(\alpha) \widehat A'_{u}(\beta) \widehat A'_{u}(\gamma) \E_{f,g,p} \left[\chi_\alpha (f) \chi_\beta (g) \chi_\gamma (p^2+1 - f\circ \pi- g) \right],
\end{equation}
where the summation is over $\alpha\in \Lambda^{m_v}_{2d}$, $\beta,\gamma\in \Lambda^{m_u}_{2d}$ and $\Lambda$ is as defined in \lref[Lemma]{lem:fourier}.
From the orthonormality of characters, the non-zero terms satisfy  $\beta = \gamma$ and $\alpha = \pi_{3}(\beta)$. Substituting in  \eqref{eq:3u3c-quv}, we get
\begin{equation}
\label{eqn:3u3c-quv-simplified}
 Q(u,v) = \sum_{\beta } \underbrace{\widehat A'_{u}(\beta)^{2}\widehat A'_{v}(\pi_{3}(\beta)) \E_{p}\left[\chi_\beta (p^2+1) \right]}_{\xi_{u,v}(\beta)} .
\end{equation}

\begin{claim}
\label{claim:empty-term}
If $\ell = 32/\delta^2$, there exists layers $0\leq i < j < \ell$ such that
$\E_{(u,v)\in E_{ij}}\left[\xi_{u,v}(0)\right] \geq \delta^5/2^9$.
\end{claim}
\begin{proof}
  Since $A'$ has fractional size $\delta$, there exists a set $S$ of
  vertices of fractional size $\delta/2$ such that $\forall v \in S,
  \widehat A'_{v}(0) =\E_f \left[A'_v(f)\right] \geq
  \delta/2$. Furthermore, there exists $\delta \ell/4$ layers, in
  which the fractional size of $S_i := S\cap V_i$ in layer $V_i$ is at
  least $\delta/4$. Since $\ell = 32/\delta^2$, we obtain from
  \lref[Theorem]{thm:layered-label-cover} that there exists
  layers $i,j$ such that the fraction of edges in $E_{ij}$ between $S_i$ and $S_j$
  is at least $\delta'=\delta^2/64$.  From above, we have that
$$\E_{(u,v)\in E_{ij}}\left[\xi_{u,v}(0)\right] \geq   \delta' \cdot(\delta/2)^3 \geq \delta^5/2^9.\qedhere$$
\end{proof}

For the rest of the proof, layers $i,j$ will be fixed as given by
\lref[Claim]{claim:empty-term}. To analyze the expression in
\eqref{eqn:3u3c-quv-simplified}, we consider the following breakup of
$\Lambda^{m_i}_{2d}\setminus\{0\}$ for every $(u,v) \in E_{ij}$: $\Far := \{\beta \in \Lambda^{m_i}_{2d} :
\Delta(\beta,(\Pe^{m_i}_{2d})^\perp) \geq 3^{d/2} \}$, $\Near_1 := \{
\beta \in \Lambda^{m_i}_{2d}\setminus \Far : \beta \neq 0 \text{ and }
\pi_3(\beta) \notin (\Pe^{m_v}_{2d})^\perp\}$ and $\Near_0 := \{ \beta \in
\Lambda^{m_i}_{2d}\setminus \Far : \beta \neq 0 \text{ and }
\pi_3(\beta) \in (\Pe^{m_v}_{2d})^\perp\}$. In Claims~\ref{claim:3c3u-far},
\ref{claim:3c3u-near0} and \ref{claim:3c3u-near1}, we bound the
absolute values of the sum of $\E_{u,v}\left[\xi_{u,v}(\beta)\right]$ for $\beta$ in $\Far,\Near_0$ and
$\Near_1$ respectively. 

\begin{claim}
\label{claim:3c3u-far}
$\left|\E_{(u,v)\in E_{ij}} \left[\sum_{\beta \in \Far}\xi_{u,v}(\beta)\right] \right| \leq \exp(-3^{\Omega(d)}) $.
\end{claim}
\begin{claim}
\label{claim:3c3u-near0}
$\left|\E_{(u,v)\in E_{ij}} \left[\sum_{\beta \in
      \Near_1}\xi_{u,v}(\beta)\right] \right| \leq 2^{-\Omega(r)}\cdot
3^d$. 
\end{claim}
\begin{claim}
\label{claim:3c3u-near1}
$\left|\E_{(u,v)\in E_{ij}} \left[\sum_{\beta \in \Near_0}\xi_{u,v}(\beta)\right] \right| \leq \eta \cdot3^d $.
\end{claim}

Combined with \lref[Claim]{claim:empty-term}, this
exhausts all terms in the expansion \eqref{eqn:3u3c-quv-simplified}. \lref[Lemma]{lem:3c3u-soundness} now follows from Claims~\ref{claim:empty-term}--\ref{claim:3c3u-near1}.
\end{proof}

We now proceed to the proofs of Claims~\ref{claim:3c3u-far},
\ref{claim:3c3u-near0} and \ref{claim:3c3u-near1}.

\begin{proof}[Proof of {\lref[Claim]{claim:3c3u-far}}]
$$\left|\E_{(u,v)\in E_{ij}} \left[\sum_{\beta \in
      \Far}\xi_{u,v}(\beta)\right] \right| \leq \E_{(u,v)\in
  E_{Ij}}\left[ \sum_{\beta \in \Far}|\widehat {A'_u}(\beta)|^2 \cdot
  |\widehat {A'_v}(\pi_3(\beta))| \cdot
  \left|\E_{p}\left[\omega^{\langle\beta,
        p^2+1\rangle}\right]\right|\right]. $$ The quantity $\langle
\beta,p^2 \rangle$ is analyzed in
\lref[Section]{sec:test-analysis}. Let $z$ be a uniformly random
$\F_3$ element. By Lemmas~\ref{lem:rank-dist} and \ref{lem:far-rank},
we get that the statistical distance between the distributions of
$\langle \beta,p^2+1\rangle$ and $z$ is $ \exp(-3^{\Omega(d)})$. Since
the $\E_z \left[\omega^z\right]=0$, we have that
$\left|\E_{p}\left[\omega^{\langle\beta, p^2+1\rangle}\right]\right|\leq
\exp(-3^{\Omega(d)})$. The claim follows since $\left|\widehat
{A'_v}(\alpha)\right|\leq 1$ for any $\alpha$ and $\sum_\beta |\widehat
{A'_u}(\beta)|^2 \leq 1$ .
\end{proof}

\begin{proof}[Proof of {\lref[Claim]{claim:3c3u-near0}}]
It suffices to bound the following for proving the claim.
\begin{align*}
&\E_{(u,v)\in E_{ij}}\left[ \sum_{\beta\in \Near_1} |\widehat
  {A'_u}(\beta)|^2\cdot |\widehat {A'_v}(\pi_3(\beta))| \right]  \\
&\leq \E_{(u,v)\in E_{ij}} \left[\sqrt{ \sum_{\beta\in \Near_1}|\widehat
    {A'_u}(\beta)|^2\cdot  |\widehat {A'_v}(\pi_3(\beta))|^2 }
  \sqrt{\sum_{\beta\in \Near_1} |\widehat
    {A'_u}(\beta)|^2}\right]\qquad [\text{ by Cauchy-Schwarz }]\\
& \leq  \sqrt{ \E_{(u,v)\in E_{ij}}  \left[\sum_{\beta\in \Near_1}
    |\widehat {A'_u}(\beta)|^2\cdot |\widehat {A'_v}(\pi_3(\beta))|^2
  \right]} \qquad [\text{ by Jensen's inequality }].
\end{align*}

We bound the above using a Fourier decoding argument as in the proof
of \lref[Claim]{clm:2c8u-near-1}.
For every vertex $v \in V_i \cup V_j$, pick a random $\beta$ according
to $|\widehat A'_v(\beta)|^2$ (note $\sum_\beta |\widehat
A'_v(\beta)|^2\leq 1$) and assign a random labeling to $v$ from the
support of $\beta$. By an argument identical to the proof of
\lref[Claim]{clm:2c8u-near-1}, we get (using the soundness of the
multilayered labelcover from \lref[Theorem]{thm:layered-label-cover}),
$$\frac{1}{3^d}\E_{(u,v)\in E_{ij}}  \left[\sum_{\beta\in \Near_1}  |\widehat {A'_v}(\pi_3(\beta))|^2 |\widehat {A'_u}(\beta)|^2\right] \leq 2^{-\Omega(r)}. \qedhere$$
\end{proof}

\begin{proof}[Proof of {\lref[Claim]{claim:3c3u-near1}}]
We bound this sum using the smoothness property of the label cover instance.

\begin{align*}
\E_{(u,v)\in E_{ij}} \left[\sum_{\beta \in \Near_0} |\widehat
  {A'_u}(\beta)|^2\cdot |\widehat {A'_v}(\pi_3(\beta))| \right] \leq
\E_{u\in V_i} \left[\sum_{\beta\notin \Far \cup \{0\}} \Pr_{v:
    (u,v)\in E_{ij}}\left[\pi_3(\beta) \in (\Pe^{m_v}_{2d})^\perp\right] \cdot  |\widehat
  {A'_u}(\beta)|^2\right].  
\end{align*}
We now argue that for every $u$ and $\beta \notin \Far \cup \{0\}$, $\Pr_{
    (u,v)\in E_{ij}}\left[\pi_3(\beta) \notin (\Pe^{m_v}_{2d})^\perp\right]$ is at most
  $3^d\cdot \eta$. This combined with the fact that $\sum_\beta
  |\widehat{A'_u}(\beta)|^2 \leq 1$ yields the claim. For every $u\in
  V_i$ and $\beta$ such that $0\neq |\supp(\beta)| =\Delta(\beta,(\Pe^{m_u}_{2d})^\perp)\leq 3^{d/2}$, by the smoothness property
  (\lref[Theorem]{thm:layered-label-cover}), we have that with
  probability at least $1-3^d\eta$, we have 
\begin{equation}\forall a \neq a' \in
  \supp(\beta), \pi(a) \neq \pi(a').\label{eq:smooth}
\end{equation}
When \eqref{eq:smooth} holds, we have $\pi_3(\beta) \neq 0$. Now since
$|\supp(\pi_3(\beta))|\leq |\supp(\beta)|\leq 3^{d/2}$ and non-zero
polynomials in $(\Pe^{m_v}_{2d})^\perp$ has support at least $3^d$, we
can further conclude that $\pi_3(\beta) \notin (\Pe^{m_v}_{2d})^\perp$
whenever \eqref{eq:smooth} holds.
\end{proof}


\begin{proof}[Proof of {\lref[Theorem]{thm:3u3c}}]
Given the completeness (\lref[Lemma]{lem:3c3u-completeness}) and
soundness (\lref[Lemma]{lem:3c3u-soundness}), we only need to fix parameters.
Let $n$ be the size of the $\TSAT$ instance and $N$ the size of the
hypergraph produced by the reduction.

Let $d = C_1\log\log (1/\delta'), \eta = (\delta')^5/C_2$ and $r =
C_3\log (1/\delta')$ for large enough constants $C_1, C_2, C_3$ and
parameter $\delta' \in (0,1)$ to be determined
shortly. By \lref[Lemma]{lem:3c3u-soundness},  if $H$ has an independent
set of size $\delta N$, then $\delta^5/2^9 \leq 3^d\cdot
2^{-\Omega(r)} + 3^d\cdot \eta + \exp(-3^{\Omega(d)}) < (\delta')^5/2^9$ for
large enough $C_1,C_2,C_3$. Hence, $H$ has no independent
sets of $\delta'N$. 

The hypergraph $H$  produced by the reduction is of size $N = \ell
n^{(1+1/\eta)\ell r}3^{((1+1/\eta)\ell r)^{O(d)}}.$
Setting $\ell  = C_4/(\delta')^2$ and $\log (1/\delta') = \Theta(\log\log
n/\log \log \log n)$, we get that $N = n^{2^{O(\log \log n / \log \log \log n)}}$.
Since $\log\log n = \Theta(\log \log N)$, we also get that $1/\delta' =
2^{\Theta(\log \log N / \log \log \log N)}$. This completes the proof
of \lref[Theorem]{thm:3u3c}.
\end{proof}


{\small
\bibliographystyle{plainurl}
\bibliography{coloring-bib}
}

\appendix

\section{Proof of {\lref[Claim]{conj:hyperplane}}}\label{app:hyperplane}

We need the following theorem due to Haramaty, Shpilka and Sudan~\cite{HaramatySS2013}.

\begin{theorem}[{\cite[Theorem~4.16, 1.7]{HaramatySS2013}}\label{thm:HSS}
  specialized to $\F_3$ and using absolute distances instead of
  fractional distances] There exists a constant $\lambda_3$ such that
  the following holds.
For  $\beta: \F_3^n \to \F_3$, let $A_1,\dots, A_K$ be hyperplanes
such that $\beta|_{A_i}$ is $\Delta_1$-close to some degree $r$ polynomial
on $A_i$. If $K > 3^{\lceil \frac{r+1}2\rceil+\lambda_3}$ and $\Delta_1 < 3^{n-r/2-2}/2$,
then $\Delta(\beta,\Pe^n_r) \leq 6\Delta_1 + 8\cdot 3^n/K$.
\end{theorem}
Setting the degree $r = 2n-2d-1$ in the above theorem implies that if
there are $K > 3^{n-d+\lambda_3}$ hyperplanes $A_1, \dots, A_K$ such
that $\beta|_{A_i}$ is $\Delta_1$-close to a degree $(2n-2d-1)$
polynmial on $A_i$, then $\Delta(\beta, \Pe^n_{2n-2d-1}) \leq
6\Delta_1 + 8\cdot 3^n/K$.

Suppose \lref[Claim]{conj:hyperplane} were false.
Then, for every nonzero $l \in \Pe^n_1$, at least one of
$\beta|_{\ell=0}$ or $\beta|_{\ell=1}$ or $\beta|_{\ell=2}$ is
$\Delta/27$-close to a degree $(2n-2d-1)$ polynomial. We thus, get
$K=(3^n-1)/2$ hyperplanes such that the restriction of $\beta$ to these
hyperplanes is $\Delta/27$-close to a degree $(2n-2d-1)$
polynomial. Observe that $K \geq 3^{n-d+\lambda_3}$ if $d \geq d_0
\geq \lambda_3+2$ and $\Delta/27 < 3^{n-(2n-2d-1)/2-2}/2 =
3^{d-1.5}/2$ if $\Delta < 3^d$. Hence, by \lref[Theorem]{thm:HSS} we
have  $\Delta(\beta, \Pe^n_{2n-2d-1}) \leq 6
\Delta/27 + 2\cdot 8\cdot 3^n/(3^n-1) < 6\Delta/27 + 32 <
\Delta$ (since $\Delta \geq 3^4$). This contradicts the hypothesis
that $\beta$ is $\Delta$-far from $\Pe^{n}_{2n-2d-1}$.
\qed

\end{document}